\documentclass{biometrika}

\usepackage{times}
\usepackage{bm}
\usepackage{natbib}

\usepackage{amsmath}
\usepackage{amssymb}
\usepackage{graphicx,subfigure}
\usepackage{pst-plot, pst-node, pst-text}
\usepackage{algorithm}
\usepackage{algorithmic}

\newcommand{\PP}{\mbox{pr}}
\newcommand{\EE}{E}
\newcommand{\R}{\mathbb{R}}

\newcommand{\Nat}{\mathbb{N}}
\newcommand{\eps}{\varepsilon}

\newcommand{\parcov}{\mbox{parcov}}
\newcommand{\Cov}{\mbox{cov}}
\newcommand{\Cor}{\mbox{cor}}
\newcommand{\Var}{\mbox{var}}
\newcommand{\peff}{\mathrm{peff}}
\newcommand{\reach}{\mathrm{reach}}

\begin{document}



\markboth{P. B\"uhlmann, M. Kalisch and M. H. Maathuis}{Partial faithfulness and the PC-simple algorithm}

\title{Variable selection in high-dimensional linear models:\\
partially faithful distributions and the PC-simple algorithm}

\author{P. B\"UHLMANN, M. KALISCH \and M. H. MAATHUIS}
\affil{Seminar f\"ur Statistik, Department of Mathematics, ETH Zurich, R\"amistrasse 101, 8092 Zurich, Switzerland \email{buhlmann@stat.math.ethz.ch} \email{kalisch@stat.math.ethz.ch} \email{maathuis@stat.math.ethz.ch}}

\maketitle

\begin{abstract}
  We consider variable selection in high-dimensional linear models where
  the number of covariates greatly exceeds the sample size. We introduce
  the new concept of partial faithfulness and use it to infer associations
  between the covariates and the response. Under partial faithfulness, we
  develop a simplified version of the PC algorithm \citep{sgs00}, the
  PC-simple algorithm, which is computationally feasible even with
  thousands of covariates and provides consistent variable selection under
  conditions on the random design matrix that are of a different nature
  than coherence conditions for penalty-based approaches like the Lasso.
   Simulations and application to real data show that
   our method is competitive compared to penalty-based approaches. We
   provide an efficient implementation of the algorithm in the \texttt{R}-package \texttt{pcalg}.
\end{abstract}


\section{Introduction}
Variable selection in high-dimensional models has recently
attracted a lot of attention. A particular stream of research
has focused on penalty-based estimators whose computation is
feasible and provably correct
\citep{mebu06,zou06,zhaoyu06,cantao07,geer07,zhanghua07,wainwright09,meyu06,huangetal06,bicketal07,wassroe07,canpan07}. Another
important approach for estimation in high-dimensional
settings, including variable selection, has been developed
within the Bayesian paradigm, see for example
\citet{geormc93,geormc97,brownetal99,brownetal02,nottko05,parkca08}. These
methods rely on Markov chain Monte Carlo techniques which are typically
very expensive for
truly high-dimensional problems.
%

In this paper, we propose a method for variable selection in linear models
which is fundamentally different from penalty-based schemes. Reasons to look at
such an approach include: (i) From a practical perspective, it is valuable
to have a very different method in the tool-kit for high-dimensional data
analysis, raising the confidence for relevance of variables if they are
selected by more than a single method. (ii) From a methodological and
theoretical perspective, we introduce the new framework of partial
faithfulness.  Partial faithfulness is related to, and typically weaker
than, the concept of linear faithfulness used in graphical models, hence
the name partial faithfulness.  We prove that partial faithfulness arises
naturally in the context of linear models if we make a
simple assumption on the structure of the regression coefficients to
exclude adversarial cases, see assumption (C2) and Theorem
\ref{theorem.mod.lin.part.faithful}.

Partial faithfulness can be exploited to construct an efficient
hierarchical testing algorithm, called the PC-simple algorithm, which is a
simplification of the PC algorithm \citep{sgs00} for estimating directed
acyclic graphs. The letters PC stand for the first names of Peter Spirtes
and Clarke Glymour, the inventors of the PC algorithm. We prove consistency of the PC-simple
algorithm for variable selection in high-dimensional partially faithful
linear models under assumptions on the design matrix that are very
different from coherence assumptions for penalty-based methods.
The PC-simple algorithm can also be viewed as a generalization of
correlation screening or sure independence screening \citep{fanlv07}. Thus,
as a special case, we obtain consistency for correlation screening under
different assumptions and reasoning than \citet{fanlv07}.
We illustrate the PC-simple algorithm, using our implementation in the
\texttt{R}-package \texttt{pcalg}, on high-dimensional simulated examples
and a real data set on riboflavin production by the bacterium
Bacillus subtilis.
%

\section{Model and notation}\label{sec.model}

Let $X = (X^{(1)},\dots,X^{(p)}) \in \R^p$ be a vector of covariates with $\EE(X) = \mu_X$ and $\Cov(X) = \Sigma_X$.
Let $\epsilon \in \R$ with $\EE(\epsilon)=0$ and $\Var(\epsilon)=\sigma^2>0$, such that $\epsilon$ is
uncorrelated with $X^{(1)}, \dots, X^{(p)}$. Let $Y\in \R$ be defined by the
following random design linear model:
\begin{eqnarray}\label{mod.lin}
  & &Y = \delta + \sum_{j=1}^p \beta_j X^{(j)} + \epsilon,
\end{eqnarray}
for some parameters $\delta \in \R$ and $\beta=(\beta_1,\ldots ,\beta_p)^T \in
\R^p$. We assume implicitly that
$\EE(Y^2) < \infty$ and $\EE\{(X^{(j)})^2\} < \infty$ for $j = 1,\ldots, p$.

We consider models in which some, or most, of the $\beta_j$s are equal to zero.
Our goal is to identify the active set
\begin{align*}
   {\cal A} = \{j=1,\ldots ,p; \ \beta_j \neq 0\}
\end{align*}
based on a sample of independent observations $(X_1,Y_1), \dots, (X_n,
Y_n)$ which are distributed as $(X,Y)$. We denote the effective dimension
of the model, that is, the number of nonzero $\beta_j$s, by $\peff = | {\cal
  A} |$.  We define the following additional conditions:
\begin{enumerate}
  \item[(C1)] $\Sigma_X$ is strictly positive definite.
\end{enumerate}
\begin{enumerate}
\item[(C2)] The regression coefficients satisfy
  \begin{eqnarray*}
     & &\{\beta_j;\ j \in {\cal A}\} \sim f(b) db,
  \end{eqnarray*}
  where $f(\cdot)$ denotes a density on a subset of $\R^{\peff}$ of an
  absolutely continuous distribution with respect to Lebesgue measure.
\end{enumerate}
Assumption (C1) is a condition on the random design matrix. It is needed
for identifiability of the regression parameters from the joint
distribution of $(X,Y)$, since $\beta = \Sigma_X^{-1}
{\Cov(Y,X^{(1)}),\ldots ,\Cov(Y,X^{(p)})}^T$.  Assumption (C2) says that
the non-zero regression coefficients are realizations from an
absolutely continuous distribution with respect to Lebesgue measure. Once
the $\beta_j$s are realized, we fix them such that they can be considered
as deterministic in the linear model \eqref{mod.lin}. This framework is
loosely related to a Bayesian formulation treating the $\beta_j$s as
independent and identically distributed random variables from a prior
distribution which is a mixture of a point mass at zero for $\beta_j$s
with $j\notin \cal A$ and a density with respect to Lebesgue measure for
$\beta_j$s with $j\in \cal A$.  Assumption (C2) is mild in the following
sense: the zero coefficients can arise in an arbitrary way and only the
non-zero coefficients are restricted to exclude adversarial
cases. Interestingly, \citet{canpan07} also make an assumption on the
regression coefficients using the concept of random sampling in their
generic S-sparse model, but other than that, there are no immediate deeper
connections between their setting and ours.  Theorem
\ref{theorem.mod.lin.part.faithful} shows that assumptions (C1) and (C2)
imply partial faithfulness, and partial faithfulness is used to obtain the
main results in Theorems
\ref{theorem.correctness.pc.simple.pop}, \ref{theorem.cons.pc.simple} and
\ref{theorem.cons.corr.screening}. Assumption (C2), however, is not a
necessary condition for these results.

We use the following notation. For
a set ${\cal S} \subseteq \{1,\dots,p\}$, $|\cal S|$ denotes its
cardinality, ${\cal S}^C$ is its complement in
$\{1,\dots,p\}$, and $X^{(\cal S)} = \{X^{(j)};\ j \in {\cal S}\}$. Moreover, $\rho(Z^{(1)},Z^{(2)} \mid W)$ and
$\parcov(Z^{(1)},Z^{(2)} \mid W)$ denote the
population partial correlation and the population partial covariance
between two variables $Z^{(1)}$ and $Z^{(2)}$ given a collection of
variables $W$.

\section{Linear faithfulness and partial faithfulness}\label{sec.faith.and.part.faith}

\subsection{Partial faithfulness}\label{sec.part.faith}

We now define partial faithfulness, the concept that will allow
us to identify the active set $\cal A$ using a simplified version of the PC
algorithm.

\begin{definition}\label{definition1}
  Let $X\in \R^p$ be a random vector, and let $Y \in \R$
  be a random variable. The distribution of $(X,Y)$ is
  said to be partially faithful if the following holds for every $j\in
  \{1,\dots,p\}$: if $\rho(Y,X^{(j)} \mid X^{(\cal S)})=0$ for some ${\cal
    S} \subseteq \{j\}^C \notag$ then $\rho(Y,X^{(j)} \mid X^{(\{j\}^C)}) =
  0$.
\end{definition}

For the linear model \eqref{mod.lin} with assumption (C1), $\beta_j = 0$ if and only if $\rho(Y,X^{(j)} \mid X^{(\{j\}^C)})=0$. Hence, such a model satisfies the partial faithfulness assumption if for every $j \in \{1,\ldots ,p\}$:
\begin{align}\label{eq: def partial faithfulness for mod.lin}
   \rho(Y,X^{(j)} \mid X^{({\cal S})}) = 0\ \mbox{for some}\ {\cal S}
  \subseteq \{j\}^C \ \mbox{implies} \ \beta_j = 0.
\end{align}

\begin{theorem}\label{theorem.mod.lin.part.faithful}
   Consider the linear model (\ref{mod.lin}) satisfying assumptions
   (C1) and (C2). Then partial faithfulness
   holds almost surely with respect to the distribution
   generating the non-zero regression coefficients.
\end{theorem}

A proof is given in the Appendix. This is in the same spirit as a result
by \citet[Th. 3.2]{sgs00} for graphical models, saying that non-faithful
distributions for directed acyclic graphs have Lebesgue measure zero, but
we are considering here the typically weaker notion of partial
faithfulness.  A direct consequence of partial faithfulness is as follows:
\begin{corollary}\label{corollary.all.par.corr.neq.0}
  Consider the linear model (\ref{mod.lin}) satisfying the partial
  faithfulness condition. Then the following holds for every $j\in
  \{1,\dots,p\}$:
  \begin{eqnarray*}
    \rho(Y,X^{(j)} \mid X^{({\cal S})}) \neq 0\ \mbox{for all}\ {\cal S} \subseteq
    \{j\}^C \ \mbox{if and only if}\ \beta_j \neq 0.
  \end{eqnarray*}
\end{corollary}
A simple proof is given in the Appendix. Corollary \ref{corollary.all.par.corr.neq.0} shows that,
under partial faithfulness, variables in the active set ${\cal A}$ have a strong
interpretation in the sense that all corresponding partial correlations are
different from zero when
conditioning on any subset ${\cal S} \subseteq \{j\}^C$.

\subsection{Relationship between linear faithfulness and partial faithfulness}\label{sec.relationship.faith.and.part.faith}

To clarify the meaning of partial faithfulness, this section discusses
the relationship between partial faithfulness and the concept of linear
faithfulness used in graphical models. This is the only section
that uses concepts from graphical modeling, and it is not required to
understand the remainder of the paper.

We first recall the definition of linear faithfulness. The distribution of a
collection of random variables $(Z^{(1)},\dots,Z^{(q)})$
can be depicted by a directed acyclic graph $G$ in which each vertex
represents a variable, and the directed edges between the vertices encode
conditional dependence relationships. The distribution of $(Z^{(1)},\dots,Z^{(q)})$
is said to be linearly faithful to $G$ if the
following holds for all $i\neq j\in \{1,\dots,q\}$ and ${\cal S}\subseteq
\{1,\dots,q\} \setminus \{i,j\}$:
\begin{align*}
    Z^{(i)} \ \text{and} \ Z^{(j)}\ \text{are d-separated by}\ Z^{(S)} \
    \text{in} \ G\ \mbox{if and only if}\ \rho(Z^{(i)}, Z^{(j)} \mid Z^{(S)})=0,
\end{align*}
see, e.g., \citet[page 47]{sgs00}. In other words, linear faithfulness to
$G$ means that all and only all zero partial correlations among the
variables can be read off from $G$ using d-separation, a
graphical separation criterion explained in detail in \cite{sgs00}.

Partial faithfulness is related to a weaker version of linear faithfulness. We say
that the distribution of $(X,Y)$, where $X \in \R^p$ is a random vector
 and $Y \in \R$ is a random variable, is
linearly $Y$-faithful to $G$ if the following holds for all $j \in
\{1,\dots,p\}$ and ${\cal S} \subseteq \{j\}^C$:
\begin{align}
   X^{(j)} \ \text{and} \ Y \ \text{are d-separated by}\ X^{(S)} \
   \text{in} \ G\ \mbox{if and only if}\ \rho(X^{(j)},Y \mid X^{(S)})=0.
    \label{eq: def partial faithfulness for DAGs}
\end{align}
Thus, linear $Y$-faithfulness to $G$ means that all and only all zero partial correlations
between $Y$ and the $X^{(j)}$s can be
read off from $G$ using d-separation, but it does not require that all
and only all zero partial correlations among the $X^{(j)}$s
can be read off using d-separation.

We now consider the relationship between linear faithfulness, linear
$Y$-faithfulness and partial faithfulness.  Linear faithfulness and linear
$Y$-faithfulness are graphical concepts, in that they link a distribution
to a directed acyclic graph, while partial faithfulness is not a graphical
concept. From the definition of linear faithfulness and linear
$Y$-faithfulness, it is clear that linear faithfulness implies linear
$Y$-faithfulness. The following theorem relates linear $Y$-faithfulness to
partial faithfulness:
\begin{theorem}\label{theorem: faithfulness and part faithfulness}
  Assume that the distribution of $(X,Y)$ is linearly $Y$-faithful to a directed acyclic graph in which $Y$ is childless. Then partial faithfulness holds.
\end{theorem}

A proof is given in the Appendix. A distribution is typically linearly $Y$-faithful
to several directed acyclic graphs. Theorem \ref{theorem: faithfulness and part faithfulness} applies if $Y$ is childless
in at least one of these graphs.

\medskip We illustrate Theorem \ref{theorem: faithfulness and part faithfulness}
by three examples. Example 1 shows a distribution where partial faithfulness does not hold.
In this case, Theorem \ref{theorem: faithfulness and part faithfulness} does not apply, because the distribution of $(X,Y)$ is not linearly $Y$-faithful to any directed acyclic graph in which $Y$ is childless.
Examples \ref{ex.2} and \ref{ex.3} show distributions where partial faithfulness does hold. In Example \ref{ex.2}, the distribution of $(X,Y)$ is linearly $Y$-faithful to a directed acyclic graph in which $Y$ is childless, and hence partial faithfulness follows from Theorem \ref{theorem: faithfulness and part faithfulness}. In Example \ref{ex.3}, the distribution of $(X,Y)$ is not linearly $Y$-faithful to any directed acyclic graph in which $Y$ is childless, showing that this is not a necessary condition for partial faithfulness.

\begin{example}\label{ex.1} Consider the following Gaussian linear model:
\begin{eqnarray}\label{eq: model ex 1}
X^{(1)} = \eps_1,\ X^{(2)} = X^{(1)} + \eps_2,\ Y = X^{(1)} - X^{(2)} +
\eps,
\end{eqnarray}
   where $\eps_1$, $\eps_2$ and $\eps$ are independent standard Normal random variables.
   This model can be represented by the linear model \eqref{mod.lin} with $\beta_1=1$ and $\beta_2=-1$. Furthermore,
   the distribution of $(X,Y) = (X^{(1)},X^{(2)},Y)$ factorizes according
   to the graph in Figure \ref{fig.ex.1}.

The distribution of $(X,Y)$ is not partially faithful, since $\rho(Y,X^{(1)} \mid \emptyset) = 0$ but $\rho(Y,X^{(1)} \mid X^{(2)})\neq 0$. Theorem \ref{theorem: faithfulness and part faithfulness} does not apply, because
the distribution of $(X,Y)$ is not linearly $Y$-faithful to any directed acyclic graph in which $Y$ is childless.
   For instance, the distribution of $(X,Y)$ is not linearly $Y$-faithful to the graph in Figure \ref{fig.ex.1}, since $\rho(X^{(1)},Y\mid \emptyset)=0$ but $X^{(1)}$ and $Y$ are not d-separated by the empty set. The zero correlation between $X^{(1)}$ and $Y$ occurs because $X^{(1)}=\eps_1$ drops out of the equation for $Y$ due to a parameter cancellation that is similar to equation \eqref{eq.parameter.cancellation} in the proof of Theorem \ref{theorem.mod.lin.part.faithful}: $Y = X^{(1)}-X^{(2)}+\eps = \eps_1-(\eps_1+\eps_2)+\eps = -\eps_2 + \eps$. The distribution of $(X,Y)$ is linearly faithful, and hence also linearly $Y$-faithful, to the graph $X^{(1)}\to X^{(2)}\leftarrow Y$, but this graph is not allowed in Theorem \ref{theorem: faithfulness and part faithfulness} because $Y$ has a child.

   Such failure of partial faithfulness can also be caused by hidden variables. To see this, consider the following
   Gaussian linear model:
   \begin{align*}
      X^{(1)}=\eps_1,\ X^{(3)}=\eps_3,\ X^{(2)}=X^{(1)}+X^{(3)}+\eps_2,\ Y=X^{(3)}+\eps,
   \end{align*}
   where $\eps_1,\eps_2,\eps_3$ and $\eps$ are independent standard Normal random variables. The distribution of $(X^{(1)},X^{(2)},X^{(3)},Y)$ factorizes according to the DAG $X^{(1)} \to X^{(2)} \leftarrow X^{(3)} \rightarrow Y$, and is linearly faithful to this DAG. Hence, the distribution of  $(X^{(1)},X^{(2)},X^{(3)},Y)$ is partially faithful by Theorem \ref{theorem: faithfulness and part faithfulness}. If, however, variable $X^{(3)}$ is hidden, so that we only observe $(X^{(1)},X^{(2)},Y)$, then the distribution of $(X^{(1)},X^{(2)},Y)$ has exactly the same conditional independence relationships as the distribution arising from \eqref{eq: model ex 1}. Hence, the distribution of $(X^{(1)},X^{(2)},Y)$ is not partially faithful.

\end{example}

\begin{example}\label{ex.2}
   Consider the following Gaussian linear model:
\begin{eqnarray*}
X^{(1)} = \eps_1,\ X^{(2)} = X^{(1)} + \eps_2,\ X^{(3)} = X^{(1)} +
\eps_3,\ X^{(4)} = X^{(2)} - X^{(3)} + \eps_4,\ Y = X^{(2)} + \eps,
\end{eqnarray*}
   where $\eps_1,\dots,\eps_4$ and $\eps$ are independent standard Normal random variables. This model can be represented by the linear model \eqref{mod.lin} with $\beta_1=\beta_3=\beta_4=0$ and $\beta_2=1$. Furthermore,
   the distribution of $(X,Y)=(X^{(1)},\dots,X^{(4)},Y)$ factorizes
   according to the graph in Figure \ref{fig.ex.2}.

The distribution of $(X,Y)$ is partially faithful, since $\rho(Y,X^{(j)} \mid X^{(\{j\}^C)})\neq 0$ only for $j=2$, and $\rho(Y,X^{(2)}  \mid  X^{(\cal S)}) \neq 0$ for any ${\cal S}\subseteq \{1,3,4\}$. In this example, partial faithfulness follows from Theorem \ref{theorem: faithfulness and part faithfulness}, since the distribution of $(X,Y)$ is linearly $Y$-faithful to the graph in Figure \ref{fig.ex.2} and $Y$ is childless in this graph.
The distribution of $(X,Y)$ is not linearly faithful to the graph in Figure \ref{fig.ex.2}, since $\Cor(X^{(1)},X^{(4)})=0$ but $X^{(1)}$ and $X^{(4)}$ are not d-separated by the empty set. Moreover,
   there does not exist any other directed acyclic graph to which the distribution of $(X,Y)$ is linearly faithful.
   Hence, this example also illustrates that linear $Y$-faithfulness is strictly weaker than linear faithfulness.
\end{example}

\begin{example}\label{ex.3}
  Consider the following Gaussian linear model:
\begin{eqnarray*}
 X^{(1)} = \eps_1,\ X^{(2)} = X^{(1)} + \eps_2,\ X^{(3)} = X^{(1)} +
 \eps_3,\ Y = X^{(2)} - X^{(3)} + \eps,
\end{eqnarray*}
  where $\eps_1$, $\eps_2$, $\eps_3$ and $\eps$ are independent standard
  Normal random variables. This model can be represented by the linear model \eqref{mod.lin} with $\beta_1=0$, $\beta_2=1$ and $\beta_3 = -1$. Furthermore, the distribution of $(X,Y) = (X^{(1)},X^{(2)},X^{(3)},Y)$
  factorizes according to the graph in Figure \ref{fig.ex.3}.

The distribution of $(X,Y)$ is partially faithful, since $\rho(Y,X^{(j)} \mid X^{(\{j\}^C)})\neq 0$ for $j\in \{2,3\}$, $\rho(Y,X^{(2)}  \mid  X^{(\cal S)}) \neq 0$ for any ${\cal S} \subseteq \{1,3\}$, and $\rho(Y,X^{(3)} \mid X^{(\cal S)}) \neq 0$ for any ${\cal S}\subseteq \{1,2\}$. However, in this case partial faithfulness does not follow from Theorem \ref{theorem: faithfulness and part faithfulness}, since the distribution of $(X,Y)$ is not linearly $Y$-faithful to the graph in Figure \ref{fig.ex.3}, since $\Cor(X^{(1)},Y)=0$ but $X^{(1)}$ and $Y$ are not d-separated by the empty set. Moreover, there does not exist any other directed acyclic graph to which the distribution of $(X,Y)$ is linearly $Y$-faithful.
\end{example}
\begin{figure}
  \vspace{0.3cm}
     \centering
     \psset{unit=5mm}
     \psset{linewidth=1pt}
     \psset{nodesep=3pt}
     \subfigure[Example \ref{ex.1}]{
         \label{fig.ex.1}
         \begin{pspicture}(-1,-.5)(4,3)
             \rput(0,3){\rnode{1}{$X^{(1)}$}}
             \rput(0,0){\rnode{2}{$X^{(2)}$}}
             \rput(3,3){\rnode{y}{$Y$}}
             \ncline{->}{1}{2}\Bput{$1$}
             \ncline{->}{2}{y}\Bput{$-1$}
             \ncline{->}{1}{y}\Aput{$1$}
         \end{pspicture}
     }\hspace{1cm}
     \subfigure[Example \ref{ex.2}]{
         \label{fig.ex.2}
         \begin{pspicture}(-1,-.5)(7,3)
             \rput(0,3){\rnode{1}{$X^{(1)}$}}
             \rput(3,3){\rnode{2}{$X^{(2)}$}}
             \rput(0,0){\rnode{3}{$X^{(3)}$}}
             \rput(3,0){\rnode{4}{$X^{(4)}$}}
             \rput(6,3){\rnode{y}{$Y$}}
             \ncline{->}{1}{2}\Aput{$1$}
             \ncline{->}{1}{3}\Bput{$1$}
             \ncline{->}{2}{4}\Aput{$1$}
             \ncline{->}{3}{4}\Aput{$-1$}
             \ncline{->}{2}{y}\Aput{$1$}
          \end{pspicture}
     }\hspace{1cm}
     \subfigure[Example \ref{ex.3}]{
         \label{fig.ex.3}
         \begin{pspicture}(-1,-.5)(4,3)
             \rput(0,3){\rnode{1}{$X^{(1)}$}}
             \rput(0,0){\rnode{3}{$X^{(3)}$}}
             \rput(3,3){\rnode{2}{$X^{(2)}$}}
             \rput(3,0){\rnode{y}{$Y$}}
             \ncline{->}{1}{2}\Aput{$1$}
             \ncline{->}{2}{y}\Aput{$1$}
             \ncline{->}{1}{3}\Bput{$1$}
             \ncline{->}{3}{y}\Aput{$-1$}
          \end{pspicture}
      }\caption{Graphical representation of the models used in Examples \ref{ex.1} - \ref{ex.3}.}
\end{figure}
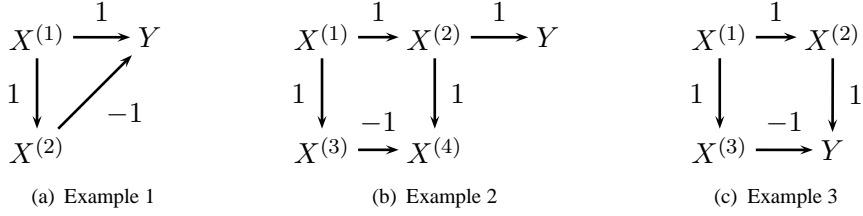

\section{The PC-simple algorithm}\label{sec.algo}
\subsection{Population version of the PC-simple algorithm}\label{sec.pop.version.pc.simple}

We now explore how partial faithfulness can be used for variable selection. In order to show the key ideas of the algorithm, we first
assume that the population partial correlations are known. In Section
\ref{sec.sample.versions.pc.simple} we consider the more realistic
situation where partial correlations are estimated from data.

First, using ${\cal S}=\emptyset$ in expression \eqref{eq: def partial faithfulness for mod.lin} yields that $\beta_j=0$ if $\Cor(Y,X^{(j)})=0$ for some $j\in \{1,\dots,p\}$. This shows that the active set $\cal A$ cannot contain any $j$ for which $\Cor(Y,X^{(j)})=0$.
Hence, we can screen all marginal correlations between pairs $(Y,X^{(j)})$, $j=1,\dots,p$, and build a first
set of candidate active variables
\begin{eqnarray}
  {\cal A}^{[1]} = \{j=1,\ldots ,p;\ \Cor(Y,X^{(j)}) \neq 0\}. \label{screen1}
\end{eqnarray}
We call this the $\mathrm{step}_1$ active set or the correlation screening
active set, and we know by partial faithfulness that
\begin{eqnarray}\label{corscreen}
  {\cal A} \subseteq {\cal A}^{[1]}.
\end{eqnarray}
Such correlation screening may reduce the dimensionality of the problem
by a substantial amount, and due to (\ref{corscreen}),
we could use other variable selection methods on the reduced set of variables ${\cal A}^{[1]}$.

Furthermore, for each $j \in {\cal A}^{[1]}$ expression \eqref{eq: def partial faithfulness for mod.lin} yields that
\begin{eqnarray}\label{screen2}
   \rho(Y,X^{(j)} \mid X^{(k)}) = 0\ \mbox{for some}\ k \in {\cal A}^{[1]} \setminus \{j\} \ \mbox{implies} \ \beta_j = 0.
\end{eqnarray}
That is, for checking whether the $j$th covariate remains in the model, we
can additionally screen all partial correlations of order one. We only
consider partial correlations given variables in the $\mathrm{step}_1$
active set ${\cal A}^{[1]}$. This is similar to what is done in the PC
algorithm, and yields an important computational reduction while still
allowing us to eventually identify the true active set $\cal A$. Thus,
screening partial correlations of order one using (\ref{screen2}) leads to
a smaller active set
\begin{eqnarray*}
  {\cal A}^{[2]} = \{j \in {\cal A}^{[1]};\ \rho(Y,X^{(j)} \mid X^{(k)}) \neq
  0\ \mbox{for all}\ k \in {\cal A}^{[1]} \setminus \{j\}\}  \subseteq {\cal A}^{[1]}.
\end{eqnarray*}
This new $\mathrm{step}_2$ active set ${\cal
  A}^{[2]}$ further reduces the dimensionality of the candidate active set, and because of \eqref{screen2} we still have that ${\cal A}^{[2]} \supseteq {\cal A}$.
We can continue
screening higher-order partial correlations, resulting in a nested sequence of
  $\mathrm{step}_m$ active sets
\begin{eqnarray}\label{screen-models}
  {\cal A}^{[1]} \supseteq {\cal A}^{[2]} \supseteq \cdots \supseteq {\cal
  A}^{[m]} \supseteq \cdots \supseteq {\cal A}.
\end{eqnarray}
A $\mathrm{step}_m$ active set ${\cal A}^{[m]}$ could be used as dimension
reduction together with any favored variable selection method in the
reduced linear model with covariates corresponding to indices in ${\cal
  A}^{[m]}$. Alternatively, we can continue the algorithm until the
candidate active set does not change
  anymore. This leads to the PC-simple algorithm, shown in pseudocode in Algorithm
  \ref{alg.PC.simple.pop}.

\begin{algorithm}[!htbp]
   \caption{The population version of the PC-simple algorithm.}
   \label{alg.PC.simple.pop}
   \begin{algorithmic}[1]
   \STATE Set $m=1$. Do correlation screening, and build the
   $\mathrm{step}_1$ active set\\
   ${\cal A}^{[1]} = \{j=1,\ldots ,p;\ \Cor(Y,X^{(j)}) \neq 0\}$ as in \eqref{screen1}.
   \REPEAT
   \STATE $m = m+1$. Construct the $\mathrm{step}_m$ active set:
   \begin{align*}
     {\cal A}^{[m]}  = \{ & j \in {\cal A}^{[m-1]}; \ \rho(Y,X^{(j)} \mid X^{({\cal S})}) \neq
     0\\
     &  \mbox{for all}\ {\cal S} \subseteq {\cal A}^{[m-1]} \setminus
     \{j\}\ \mbox{with}\ |{\cal S}| = m-1\}.
   \end{align*}
   \UNTIL $|{\cal A}^{[m]}| \le m$.
   \end{algorithmic}
\end{algorithm}

The value $m$ that is reached in Algorithm
\ref{alg.PC.simple.pop} is called $m_{\reach}$:
\begin{eqnarray}\label{reach}
   m_{\reach} = \min\{m;\ |{\cal A}^{[m]}| \le m\}.
\end{eqnarray}
The following theorem shows correctness of the population version of the
PC-simple algorithm.
\begin{theorem}\label{theorem.correctness.pc.simple.pop}
   For the linear model (\ref{mod.lin}) satisfying (C1) and partial
   faithfulness, the
   population version of the PC-simple algorithm identifies the true underlying
   active set, i.e., ${\cal
     A}^{[m_{\reach}]} = {\cal A} = \{j=1,\ldots ,p;\ \beta_j \neq 0\}$.
\end{theorem}

A proof is given in the Appendix. Theorem \ref{theorem.correctness.pc.simple.pop} shows that
partial faithfulness, which
is often weaker than linear faithfulness, is sufficient to guarantee
correctness of the population PC-simple algorithm.
%
The PC-simple algorithm is similar to the PC algorithm
\citep[Section 5$\cdot$4$\cdot$2]{sgs00}, but there are two important differences.
First, the PC algorithm considers all ordered pairs of variables in
$(X^{(1)},\dots,X^{(p)},Y)$, while we only consider ordered pairs $(Y,
X^{(j)})$, $j\in \{1,\dots,p\}$, since we are
only interested in associations between $Y$ and $X^{(j)}$. Second, the PC algorithm
considers conditioning sets in the neighborhoods of
both $Y$ and $X^{(j)}$, while we only consider conditioning
sets in the neighborhood of $Y$.


\subsection{Sample version of the PC-simple algorithm}\label{sec.sample.versions.pc.simple}

For finite samples, the partial correlations must be estimated.
We use the following shorthand notation:
\begin{align*}
  \begin{array}{ll}
     \rho(Y,j  \mid  {\cal S}) = \rho(Y,X^{(j)} \mid X^{(\cal S)}),  & \quad \hat \rho(Y,j  \mid  {\cal S}) = \hat{\rho}(Y,X^{(j)} \mid X^{(\cal S)}),\\
     \rho(i,j  \mid  {\cal S}) = \rho(X^{(i)},X^{(j)} \mid X^{(\cal S)}), & \quad \hat \rho(i,j  \mid  {\cal S}) = \hat{\rho}(X^{(i)},X^{(j)} \mid X^{(\cal S)}),
  \end{array}
\end{align*}
where the hat-versions denote sample partial correlations. These can be
calculated recursively, since for any $k \in {\cal S}$ we have
\begin{eqnarray*}
\hat \rho(Y,j  \mid  {\cal S}) = \frac{ \hat \rho(Y,j  \mid  {\cal S} \setminus \{k\}) -
    \hat \rho(Y,k  \mid  {\cal S} \setminus \{k\})
    \hat \rho(j,k  \mid  {\cal S} \setminus \{k\})}{[{\{1 - \hat \rho(Y,k  \mid  {\cal S}
    \setminus \{k\})^2\}\{1 - \hat \rho(j,k  \mid  {\cal S} \setminus
    \{k\})^2\}}]^{1/2}}.
\end{eqnarray*}
In order to test whether a partial correlation is zero, we apply Fisher's
$Z$-transform
\begin{eqnarray}\label{ztrans}
Z(Y,j \mid {\cal S}) = \frac{1}{2} \log \left \{\frac{1 +
    \hat \rho(Y,j \mid {\cal S})}{1 - \hat \rho(Y,j \mid {\cal S})} \right\}.
\end{eqnarray}
Classical decision theory in the Gaussian case yields the following
rule. Reject the null-hypothesis $H_0(Y,j \mid {\cal S}):\ \rho(Y,j \mid
{\cal S}) = 0$ against the two-sided alternative $H_A(Y,j \mid {\cal S}):\
\rho(Y,j \mid {\cal S}) \neq 0$ if $(n-|{\cal S}|-3)^{1/2} |Z(Y,j \mid {\cal
  S})| > \Phi^{-1}(1 - \alpha/2)$, where $\alpha$ is the significance level
and $\Phi(\cdot)$ is the standard Normal cumulative distribution
function. Even in the absence of Gaussianity, the rule above is a
reasonable thresholding operation.

The sample version of the PC-simple algorithm is obtained by replacing the
statements about $\rho(Y,X^{(j)} \mid X^{({\cal S})}) \neq 0$ in Algorithm
\ref{alg.PC.simple.pop} by
\begin{eqnarray*}
  (n-|{\cal S}|-3)^{1/2} |Z(Y,j \mid {\cal S})| > \Phi^{-1}(1 -
  \alpha/2).
\end{eqnarray*}
The resulting estimated set of variables is denoted by $\widehat{\cal
  A}(\alpha) = \widehat{\cal A}^{\hat{m}_{\reach}}(\alpha)$, where
$\hat{m}_{\reach}$ is the estimated version of the quantity in
(\ref{reach}). The only tuning parameter $\alpha$ of the PC-simple algorithm is the significance
level for testing the partial correlations.

The PC-simple algorithm is very different from a greedy scheme, since
it screens many correlations or partial correlations at once
and may delete many variables at once. Furthermore, it is a more sophisticated
pursuit of variable screening than the marginal correlation approach in
\citet{fanlv07} or the low-order partial correlation method in
\citet{willepb06}. \citet{castrov06} extended the latter and considered a
limited-order partial correlation approach. However, their method does not
exploit the clever trick
of the PC-simple algorithm that it is sufficient to consider only
conditioning sets ${\cal S}$ which
have survived in the previous $\mathrm{step}_{m-1}$ active set ${\cal
  A}^{[m-1]}$. Therefore, the algorithm of \citet{castrov06} is often
infeasible and must be approximated by a Monte Carlo approach.

Since the PC-simple algorithm is a simplified version of the PC algorithm,
its computational complexity is bounded above by that of the PC algorithm.
This is difficult to evaluate exactly, but a crude bound is $O(p^{\peff})$,
see \citet[formula (4)]{kabu07}.  Section \ref{sec.numerical} shows that we
can easily use the PC-simple algorithm in problems with thousands of
covariates.

\section{Asymptotic results in high dimensions}\label{sec.asymptotics}

\subsection{Consistency of the PC-simple algorithm}

We now show that the PC-simple algorithm is consistent for variable
selection, even if $p$ is much larger than $n$. We consider the linear
model in \eqref{mod.lin}. To capture high-dimensional behavior, we let the
dimension grow as a function of sample size and thus, $p = p_n$ and also
the distribution of $(X,Y)$, the regression coefficients $\beta_j =
\beta_{j,n}$, and the active set ${\cal A} = {\cal A}_n$ with $\peff =
\peff_n = |{\cal A}_n|$ change with $n$. Our assumptions are as follows:
%
\begin{enumerate}
\item[(D1)] The distribution $P_n$ of $(X,Y)$ is multivariate Normal and satisfies (C1) and the partial faithfulness condition for all $n$.
\item[(D2)] The dimension $p_n = O(n^a)$ for some $0 \le a < \infty$.
\item[(D3)] The cardinality of the active set $\peff_n = |{\cal A}_n| =
  |\{j=1,\ldots ,p_n;\ \beta_{j,n} \neq 0\}|$ satisfies: $\peff_n = O(n^{1
    - b})$ for some $0 < b \le 1$.
\item[(D4)] The partial correlations $\rho_n(Y,j \mid {\cal S}) = \rho(Y,X^{(j)} \mid X^{({\cal S})})$ satisfy:
\begin{eqnarray*}
  & &\inf\Big\{|\rho_n(Y,j \mid {\cal S})|;\ j=1,\ldots ,p_n,\ {\cal S}
  \subseteq \{j\}^C,\ |{\cal S}| \le \peff_n\ \mbox{with}\
  \rho_n(Y,j \mid {\cal S}) \neq 0\Big\} \ge c_n,
\end{eqnarray*}
where $c_n^{-1} = O(n^{d})$ for some $0 \le d < b/2$, and $b$ is as in (D3).
\item[(D5)] The partial correlations $\rho_n(Y,j \mid {\cal S})$ and
  $\rho_n(i,j \mid {\cal S}) = \rho(X^{(i)},X^{(j)} \mid X^{({\cal S})})$ satisfy:
\begin{eqnarray*}
& &(i)\ \sup_{n,j,{\cal S} \subseteq \{j\}^C,|{\cal S}| \le \peff_n}
|\rho_n(Y,j \mid {\cal S})| \le M < 1,\\
& &(ii)\ \sup_{n,i\neq j,{\cal S} \subseteq
  \{i,j\}^C,|{\cal S}| \le \peff_n} |\rho_n(i,j \mid {\cal S})| \le M < 1.
\end{eqnarray*}
\end{enumerate}
Assumption (D1) is made to simplify asymptotic calculations, and it is not
needed in the population case.
Unfortunately, it is virtually impossible to check assumptions (D1)-(D5)
in practice, with the exception of (D2). However, this is common to
assumptions for high-dimensional variable selection, such as the
neighborhood stability condition \citep{mebu06}, the irrepresentable condition
\citep{zhaoyu06}, or the restrictive eigenvalue assumption
\citep{bicketal07}.
A more detailed discussion of assumptions (D1)-(D5) is given in Section
\ref{subsec.cond}.

Letting $\widehat{\cal A}_n(\alpha)$ denote the estimated set of variables from
the PC-simple algorithm
in Section \ref{sec.sample.versions.pc.simple} with significance level $\alpha$, we obtain the following
consistency result:
\begin{theorem}\label{theorem.cons.pc.simple}
  Consider the linear model (\ref{mod.lin}) and assume
  (D1)-(D5). Then there exists a sequence $\alpha_n \to 0\ (n \to \infty)$ and a constant $C>0$
  such that the PC-simple algorithm satisfies
\begin{eqnarray*}
  \PP\{\widehat{\cal A}_n(\alpha_n) = {\cal A}_n\}
  = 1 - O\{\exp(-Cn^{1 - 2d})\} \to 1\ (n \to
  \infty),
\end{eqnarray*}
where $d$ is as in (D4).
\end{theorem}
A proof is given in the Appendix. The value $\alpha_n$, although being
the significance level of a single test, is a tuning parameter which allows to
control type I and II errors over the many tests which are pursued in the
PC-simple algorithm. A possible choice yielding consistency is
$\alpha_n = 2\{1 - \Phi(n^{1/2} c_n/2)\}$. This choice depends on the unknown lower
bound of the partial correlations in (D4).

\subsection{Discussion of the conditions of Theorem
\ref{theorem.cons.pc.simple}}\label{subsec.cond}

There is much recent work on high-dimensional and
computationally tractable variable selection, most of it considering
versions of the Lasso \citep{tibs96} or the Dantzig
selector \citep{cantao07}. Neither of these methods exploit partial
faithfulness. Hence, it is interesting to discuss our conditions with a view towards these
other established results.

For the Lasso, \citet{mebu06} proved that a so-called neighborhood stability
condition is sufficient and almost necessary for consistent variable
selection, where the word almost refers to the fact that a strict inequality
with the relation $<$ appears in the sufficient condition whereas
for necessity, there is a $\le$
relation. \citet{zou06} and \citet{zhaoyu06} gave a different, but equivalent
condition. In the latter work, it is called the irrepresentable
condition. The adaptive Lasso \citep{zou06} or other two-stage
Lasso and thresholding procedures \citep{meyu06} yield consistent variable
selection under weaker conditions than the neighborhood
stability or irrepresentable condition, see also Example \ref{ex.lasso.inconsistent}
below. Such two-stage procedures rely on bounds for $\|\hat{\beta} -
\beta\|_q\ (q=1,2)$ whose convergence
rate to zero is guaranteed under possibly weaker restricted eigenvalue
assumptions on the design \citep{bicketal07} than what is required by the
irrepresentable or neighborhood stability condition. All these different
assumptions are not
directly comparable with our conditions (D1)-(D5).

Assumption (D2) allows for an arbitrary polynomial growth of dimension as a
function of sample size, while (D3) is a
sparseness assumption in terms of the number of effective variables. Both
(D2) and (D3) are fairly standard assumptions in high-dimensional
asymptotics. More critical are the partial faithfulness requirement in
(D1), and the conditions on the partial correlations in (D4) and (D5).

We interpret these conditions with respect to the design $X$ and the
conditional distribution of $Y$ given $X$. Regarding
the random design, we assume (C1) and (D5,(ii)). Requiring (C1) is
rather weak, since it does not impose constraints on the behavior of the
covariance matrix $\Sigma_X = \Sigma_{X;n}$ in the sequence of
distributions $P_n\ (n \in \Nat)$, except for strict positive definiteness
for all $n$. Condition (D5,(ii)) excludes perfect collinearity, where the fixed upper
bound on partial correlations places some additional restrictions on
the design.
Regarding the conditional distribution of $Y$ given $X$, we
require partial faithfulness. This
becomes more explicit by invoking Theorem
\ref{theorem.mod.lin.part.faithful}: partial faithfulness follows by
assuming condition (C2) in Section \ref{sec.model} for every $n$, which
involves the regression coefficients only. Conditions (D4) and (D5,(i))
place additional restrictions on
both the design $X$ and the conditional distribution of $Y$ given $X$.

Assumption (D4) is used for controlling the type II errors in
the many tests of the PC-simple algorithm, see the proof of Theorem
\ref{theorem.cons.pc.simple}. This assumption is slightly stronger
than requiring that all non-zero
regression coefficients are larger than a detectability-threshold, which
has been previously used for
analyzing the Lasso in \citet{mebu06}, \citet{zhaoyu06} and
\citet{meyu06}. Clearly, assumptions on the design $X$ are not sufficient
for consistent variable selection with any method and some additional
detectability assumption is needed. Our assumption (D4) is restrictive, as
it does not allow small non-zero low-order partial
correlations. Near partial faithfulness \citep{robins03}, where small partial
correlations would imply that corresponding regression
coefficients are small, would be a more realistic framework in
practice. However, this would make the theoretical arguments much more involved, and
we do not pursue this in this paper.

Although our assumptions
are not directly comparable to the neighborhood stability or
irrepresentable condition for the Lasso,
it is easy to construct examples where the Lasso fails to be
consistent while the PC-simple algorithm recovers the true set of variables, as
shown by the following example.

\begin{example}\label{ex.lasso.inconsistent}
    Consider a Gaussian linear model as in (\ref{mod.lin}) with $p = 4$, $\peff = 3$, $\sigma^2=1$, $\mu_X = (0,\dots,0)^T$
    \begin{eqnarray*}
        & &\Sigma_X = \left( \begin{array}{cccc}
            1 & \rho_1 & \rho_1 & \rho_2 \\
            \rho_1 & 1 & \rho_1 & \rho_2 \\
            \rho_1 & \rho_1 & 1 & \rho_2 \\
            \rho_2 & \rho_2 & \rho_2 & 1
        \end{array}
        \right),\ \ \rho_1 = -0.4,\ \rho_2 = 0.2,\\
    \end{eqnarray*}
where $\beta_1$, $\beta_2$, $\beta_3$ are fixed i.i.d. realizations from
${\cal N}(0,1)$ and $\beta_4 = 0$.

    It is shown in \citet[Cor. 1]{zou06} that the Lasso is inconsistent for variable
    selection in this model. On the other hand, (D1) holds because of
    Theorem \ref{theorem.mod.lin.part.faithful}, and also (D5)
    is true. Since the dimension $p$ is fixed, (D2), (D3) and (D4) hold automatically. Hence, the
    PC-simple algorithm is consistent for variable selection. It should be noted
    though that the adaptive Lasso is also consistent for this example.
\end{example}
We can slightly modify Example \ref{ex.lasso.inconsistent} to make it high-dimensional.
Consider $\peff = 3$ active variables, with
design and coefficients as in Example \ref{ex.lasso.inconsistent}. Moreover, consider $p_n -
\peff$ noise covariates which are independent from the active variables, with $p_n$ satisfying (D2).
Let the design satisfy (C1) and (D5,(ii)), for example by taking the noise covariates to be mutually independent.
Then assumptions (D1)-(D5) hold, implying
consistency of the PC-simple algorithm, while the Lasso is inconsistent.

\subsection{Asymptotic behavior of correlation screening}

Correlation screening is equivalent to sure independence
screening of \citet{fanlv07}, but our assumptions and reasoning via partial faithfulness are very different
from the work of Fan \& Lv. Denote by $\widehat{\cal A}_n^{[1]}(\alpha)$ the correlation screening
active set, estimated from data, using significance level $\alpha$, obtained from the first step of the sample version of the PC-simple algorithm. We do not require any sparsity conditions for consistency. We define:
\begin{enumerate}
\item[(E1)] as assumption (D4) but for marginal correlations
  $\Cor(Y,X^{(j)}) = \rho_n(Y,j)$ only.
\item[(E2)] as assumption (D5) but for marginal correlations
  $\Cor(Y,X^{(j)}) = \rho_n(Y,j)$ only.
\end{enumerate}

\begin{theorem}\label{theorem.cons.corr.screening}
Consider the linear model (\ref{mod.lin}) and assume
(D1), (D2), (E1) and (E2). Then there exists a sequence $\alpha_n \to 0\ (n
\to \infty)$ and a constant $C>0$
such that:
\begin{eqnarray*}
\PP\{\widehat{\cal A}_n^{[1]}(\alpha_n) \supseteq {\cal
  A}_{n}\}
= 1 - O\{\exp(-Cn^{1 - 2d})\} \to 1\ (n \to \infty),
\end{eqnarray*}
where $d >0$ is as in (E1).
\end{theorem}
A proof is given in the Appendix.
A possible choice for $\alpha_n$ is
$\alpha_n = 2\{1- \Phi(n^{1/2} c_n/2)\}$.
As pointed out above, we do not make any assumptions on
sparsity. However, for non-sparse problems, many correlations may be
non-zero and $\widehat{\cal A}^{[1]}$ can still be large, for example
almost as large as the full set $\{1,\ldots ,p\}$.

Under some restrictive conditions on the covariance $\Sigma_X$ of the random
design, \citet{fanlv07} have shown that correlation screening, or sure
independence screening, is overestimating the active set ${\cal A}$, as
stated in Theorem \ref{theorem.cons.corr.screening}.
%
%
Theorem \ref{theorem.cons.corr.screening} shows that this result
also holds under very different assumptions on $\Sigma_X$ when partial
faithfulness is assumed in addition. 
Hence, our result justifies
correlation screening as a more general tool than what it appears to be from the setting of
\citet{fanlv07}, thereby extending the range of applications.

\section{Numerical results}\label{sec.numerical}

\subsection{Analysis for simulated data}\label{sec.ROC}

We simulate data according to a Gaussian linear model as in (\ref{mod.lin})
with $\delta = 0$ and $p$ covariates with $\mu_X=(0,\dots,0)^T$ and
covariance matrix
$\Sigma_{X;i,j} = \rho^{|i-j|}$, where $\Sigma_{X;i,j}$ denotes the $(i,j)$th entry of $\Sigma$. In order to
generate values for $\beta$, we follow
(C2): a certain number $\peff$ of coefficients
$\beta_j$ have a value different from zero. The values of the nonzero
$\beta_j$s
are sampled independently from a standard normal distribution and the indices
of the nonzero $\beta_j$s are evenly spaced between $1$ and $p$.
We consider two settings:
\begin{description}
\item[] low-dimensional: $p=19$, $\peff=3$, $n=100$; $\rho \in
  \{$0,0$\cdot$3,0$\cdot$6$\}$ with $1000$ replicates
\item[] high-dimensional: $p=499$, $\peff=10$, $n=100$; $\rho \in
  \{$0,0$\cdot$3,0$\cdot$6$\}$ with $300$ replicates
\end{description}

We evaluate the performance of the methods using receiver operating
characteristic curves which measure
the accuracy for variable selection independently from
the issue of choosing good tuning parameters. We compare the PC-simple
algorithm to the Lasso \citep{efron04lars} and Elastic Net \citep{zouhastie05}, using the \texttt{R}-packages
\texttt{pcalg}, \texttt{lars} and \texttt{elasticnet}, respectively.
For Elastic Net, we vary the $\ell^1$-penalty parameter only while keeping the $\ell^2$-penalty parameter
fixed at the default value from the \texttt{R}-package.

In the low-dimensional settings shown in Figures
\ref{fig:p19r0}, \ref{fig:p19r03}, \ref{fig:p19r06}, the
PC-simple algorithm clearly dominates the Lasso and Elastic Net for small false
positive rates, which are desirable in many applications. When focusing on the false positive rate arising from the default value for
$\alpha$=0$\cdot$05 in the PC-simple algorithm, indicated by the vertical lines, the PC-simple algorithm
outperforms the Lasso and Elastic Net by a large
margin. If the correlation among the covariates increases, the performance of Elastic Net deteriorates,
whereas the performances of the PC-simple algorithm and the Lasso do not vary much.

In the high-dimensional settings shown in Figures \ref{fig:p499r0}, \ref{fig:p499r03}, \ref{fig:p499r06}, the difference between the methods is small
for small false positive rates. The Lasso performs best, Elastic Net is worst, and the PC-simple algorithm is somewhere in
between. For larger false positive rates, the differences become more pronounced. Up to the false positive rate corresponding to the default value of $\alpha$=0$\cdot$05, the PC-simple algorithm is never significantly
outperformed by either the Lasso or Elastic Net.

We refer to the working paper ``Stability selection'' by N. Meinshausen and P. B\"uhlmann for a more principled way to choose the tuning parameter $\alpha$.
Further examples, with $p=1000,\ \peff =5,\ n=50$ and equi-correlated design
$\Sigma_{X;i,j}$ = 0$\cdot$5 for $i \neq j$ and $\Sigma_{X;i,i} = 1$ for all $i$,
are reported in \citet{pb08}.

%
%

The computing time of the PC-simple algorithm on 10 different values of
$\alpha$ has about the same order of magnitude as the Lasso or Elastic Net
for their whole solution paths. Hence, the
PC-simple algorithm is certainly feasible for high-dimensional problems.

\begin{figure}[!htp]
     \centering
     \subfigure[Low dimensional, $\rho=0$.]{
           \label{fig:p19r0}
          \includegraphics[width=.25\textwidth,angle=-90]{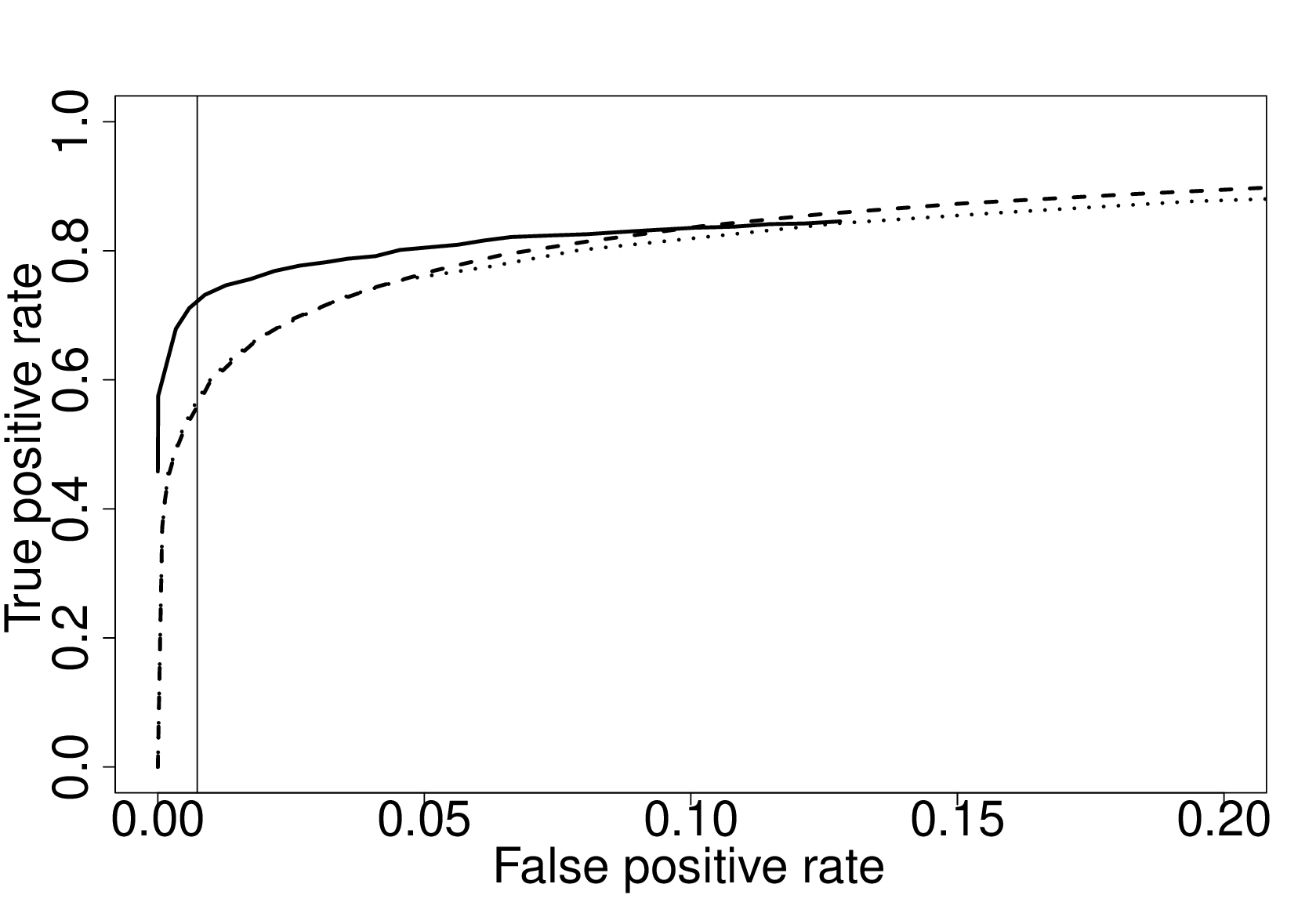}}
     \subfigure[High dimensional, $\rho=0$.]{
          \label{fig:p499r0} \includegraphics[width=.25\textwidth,angle=-90]{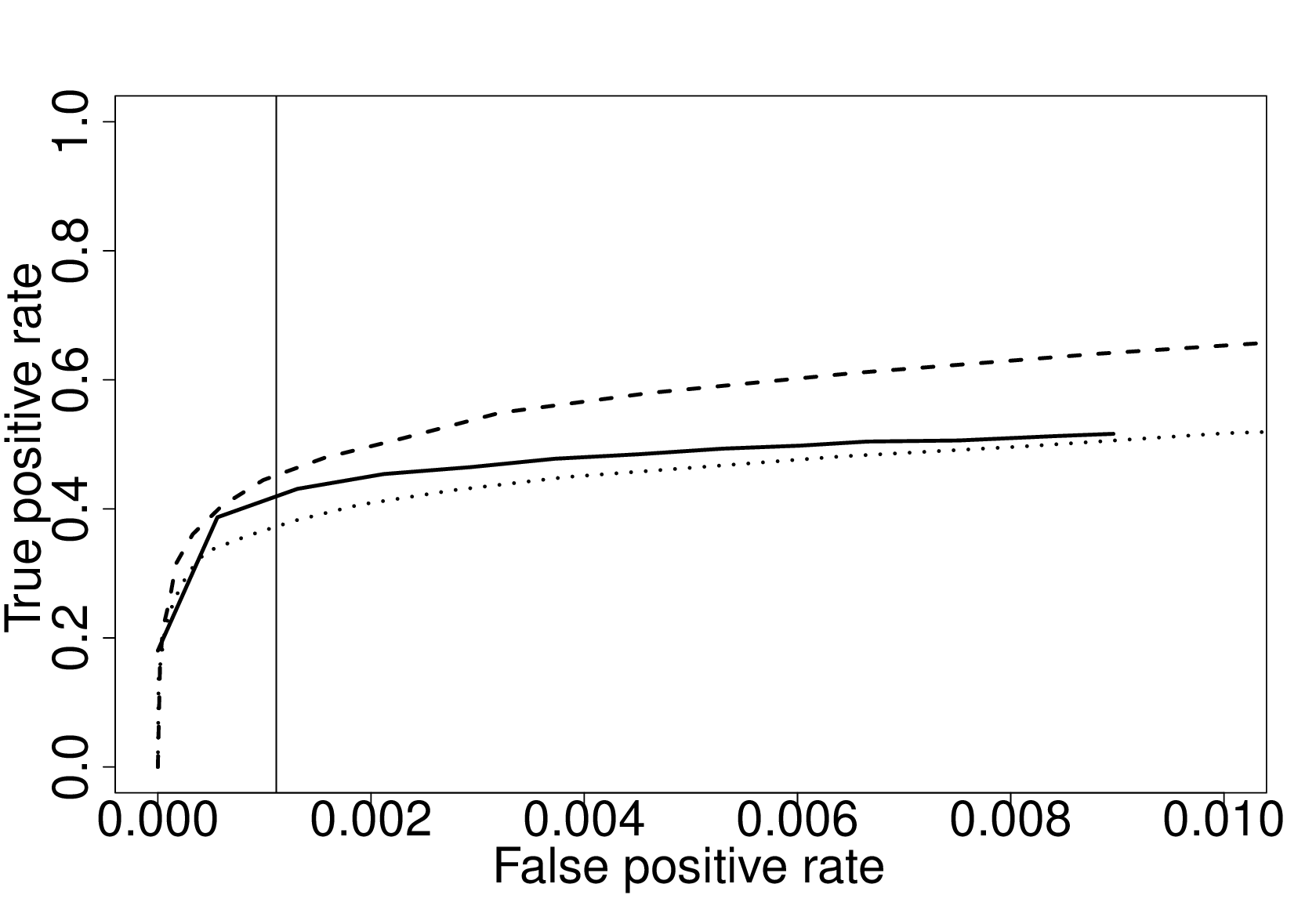}}
     \vspace{.1in}
     \subfigure[Low dimensional, $\rho$=0$\cdot$3.]{
           \label{fig:p19r03}
           \includegraphics[width=.25\textwidth,angle=-90]{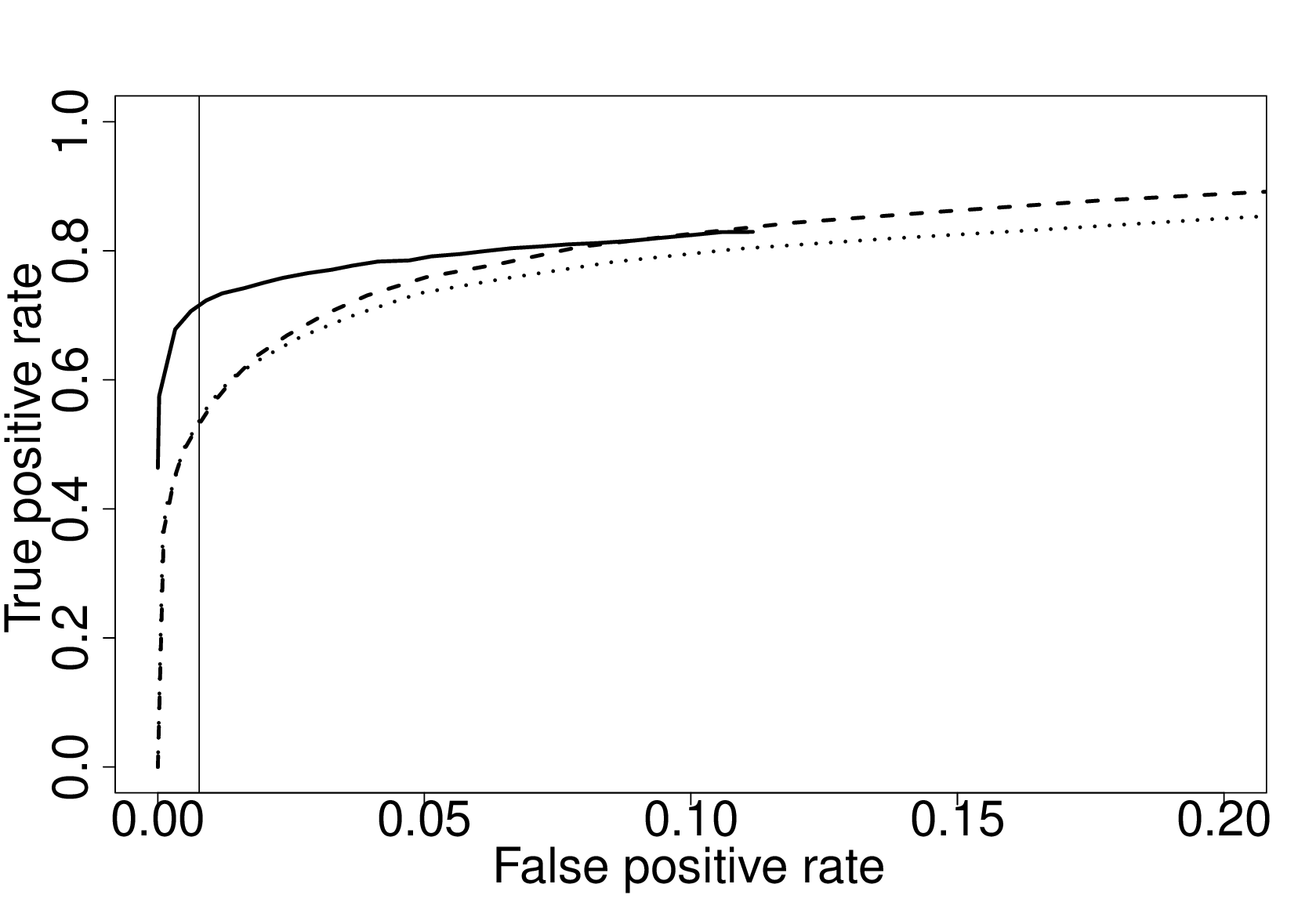}}
     \subfigure[High dimensional, $\rho$=0$\cdot$3.]{
           \label{fig:p499r03} \includegraphics[width=.25\textwidth,angle=-90]{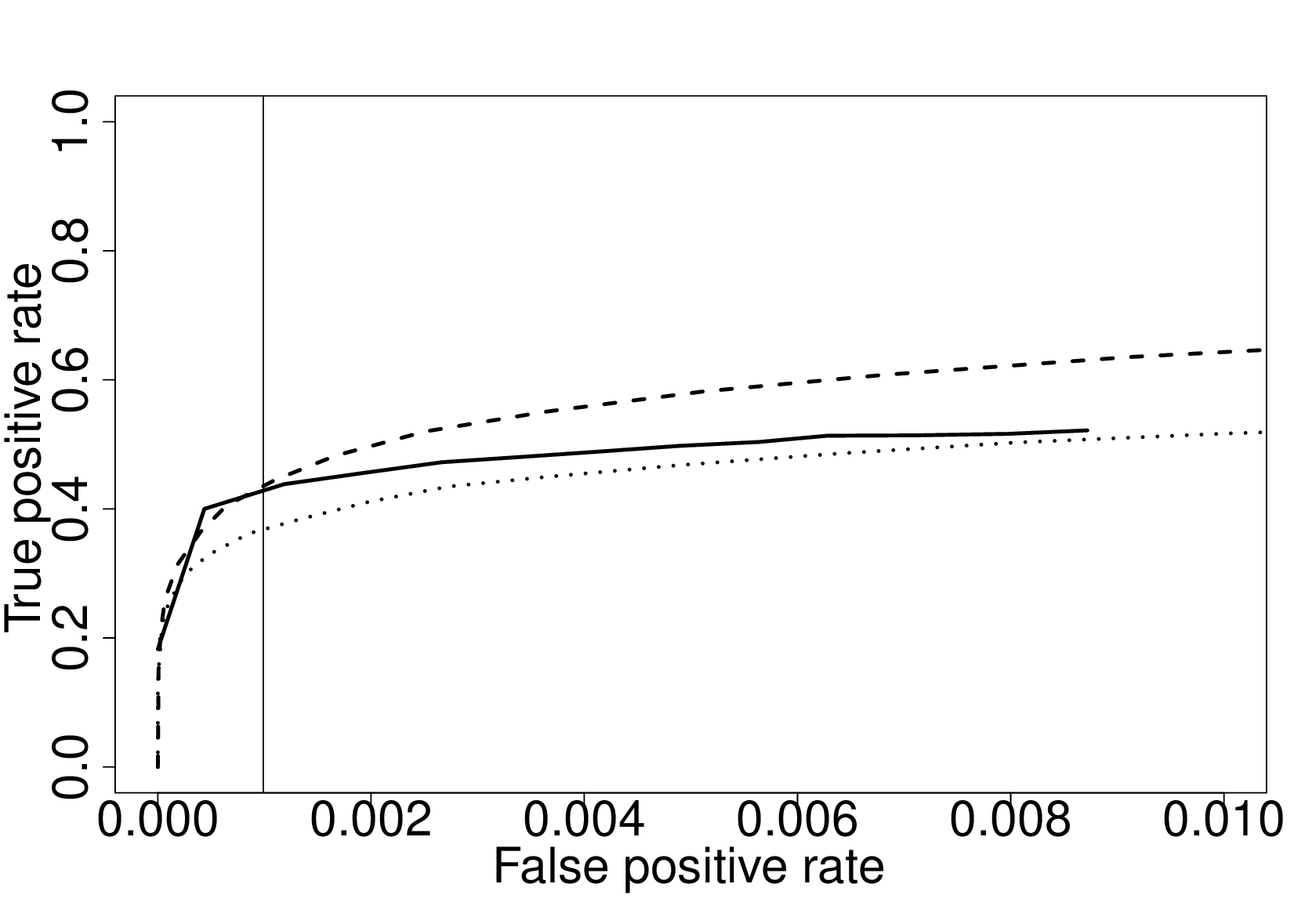}}
     \vspace{.1in}
     \subfigure[Low dimensional, $\rho$=0$\cdot$6.]{
           \label{fig:p19r06}
           \includegraphics[width=.25\textwidth,angle=-90]{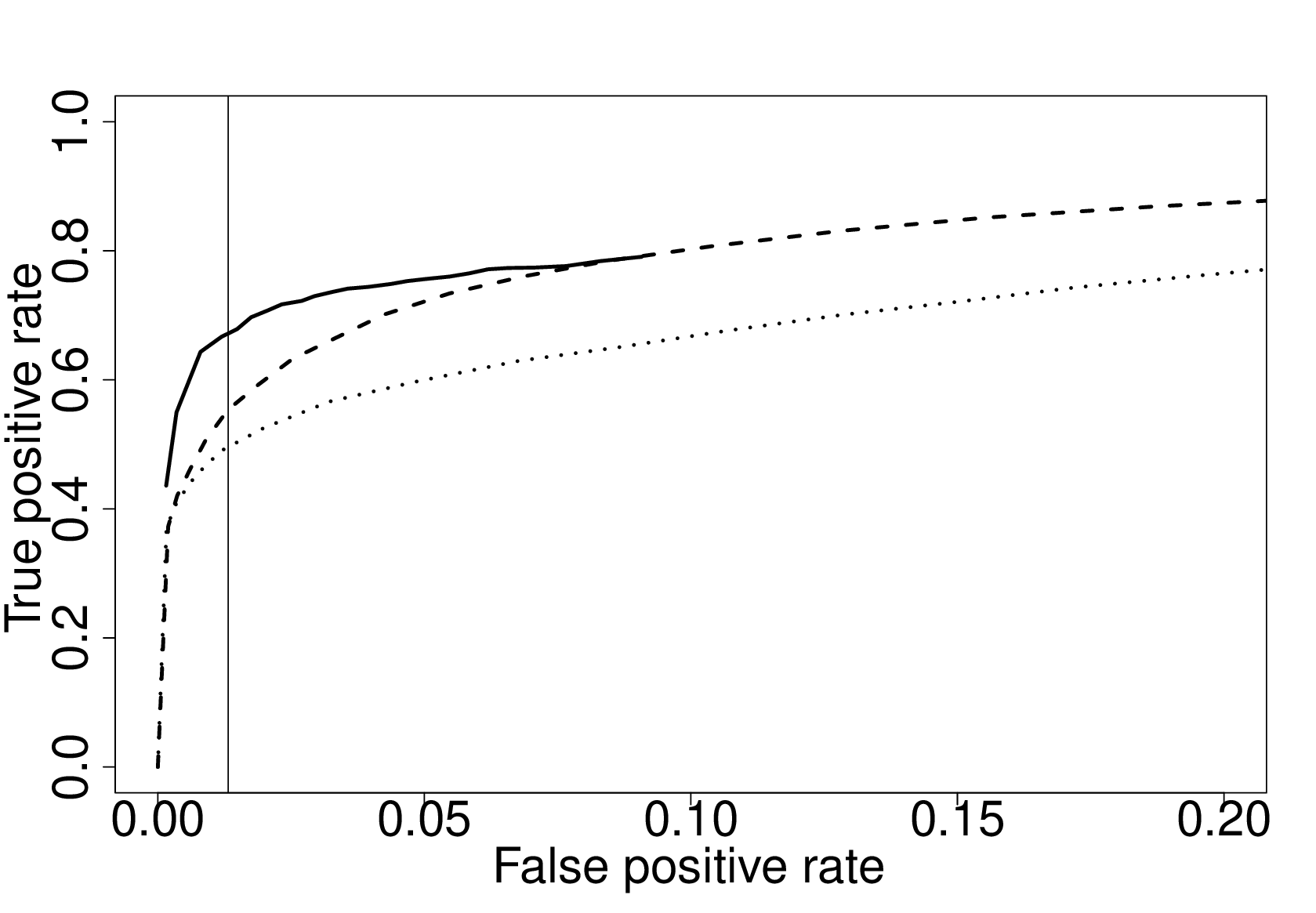}}
     \subfigure[High dimensional, $\rho$=0$\cdot$6.]{
           \label{fig:p499r06}
           \includegraphics[width=.25\textwidth,angle=-90]{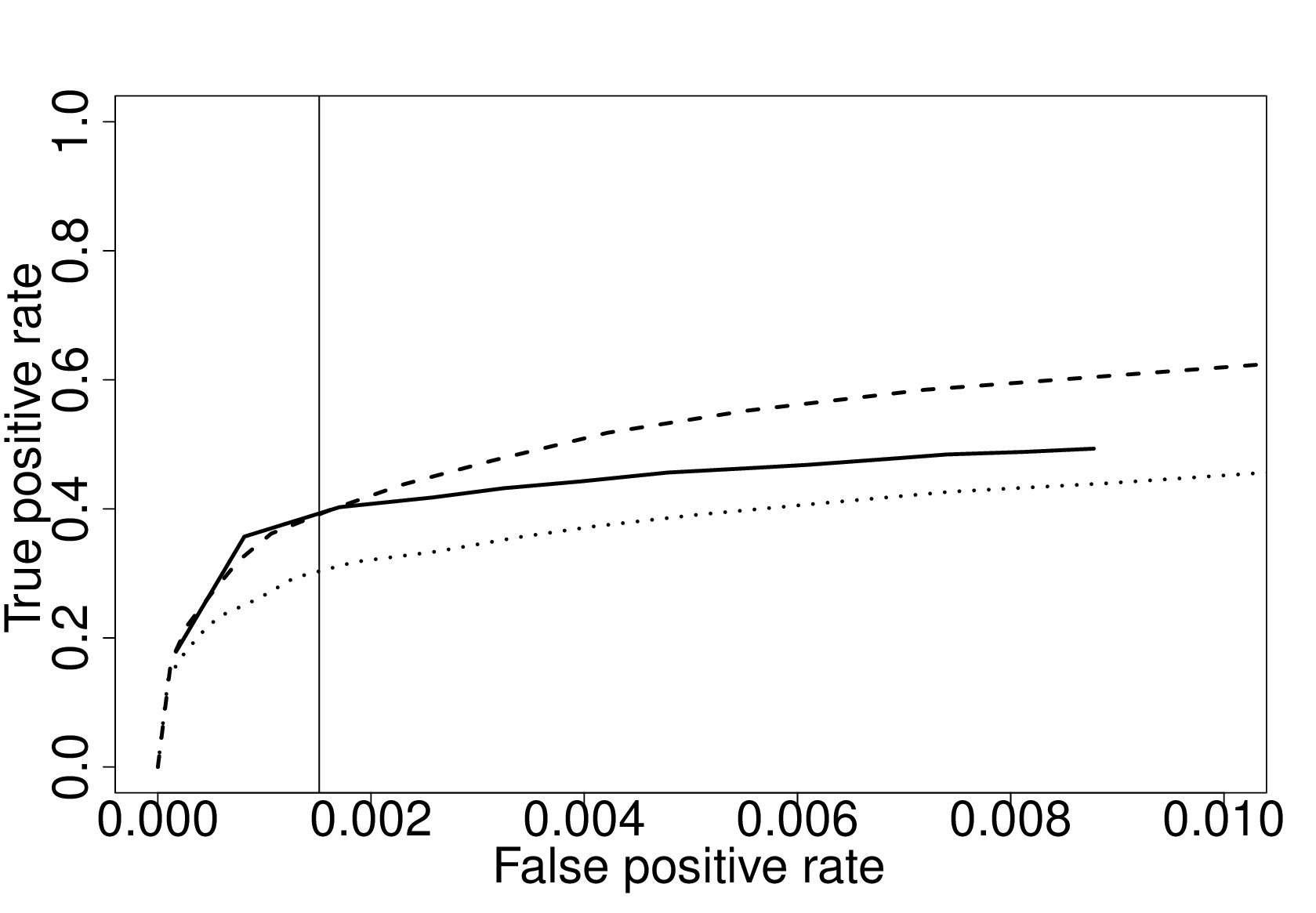}}
         \caption{Receiver operating characteristic curves for the
           simulation study in Section \ref{sec.ROC}. The solid line
           corresponds to the PC-simple algorithm, the dashed line to the
           Lasso and the dotted line to Elastic Net.  The solid vertical
           lines indicate the performance of the PC-simple algorithm using
           the default $\alpha$=0$\cdot$05.}
     \label{fig:multifig}
\end{figure}

\subsection{Prediction optimal tuned methods for simulated data}\label{sec.opt.pred.tuned}

We now compare the PC-simple algorithm to several existing methods when using prediction optimal tuning. It
is known that the prediction-optimal tuned Lasso overestimates the true
model \citep{mebu06}. The adaptive Lasso \citep{zou06} and the relaxed
Lasso \citep{me07} correct Lasso's overestimating
behavior and prediction-optimal tuning for these methods yields a good
amount of regularization for variable selection.

We simulate from a Gaussian linear model as in (\ref{mod.lin}) with $p =
1000,\ \peff = 20$, $n=100$, and
\begin{eqnarray*}
   & &\delta = 0,\ \mu_X = (0,\dots,0)^T,\ \Sigma_{X;i,j} = \textrm{0$\cdot$5}^{|i-j|},\ \sigma^2 = 1,\\
   & &\beta_1,\ldots ,\beta_{20}\ \mbox{i.i.d.}\ \sim {\cal N}(0,1),\ \
   \beta_{21} = \ldots = \beta_{1000} = 0,
\end{eqnarray*}
using 100 replicates. We consider the following performance measures:
\begin{align}\label{perf.meas}
  \begin{array}{lr}
    \|\hat{\beta} - \beta\|_2^2 = \sum_{j=1}^p (\hat{\beta}_j - \beta_j)^2 &
     (\mbox{MSE Coeff})\\
   \EE_X [ \{X^T (\hat{\beta} - \beta)\}^2] = (\hat{\beta} - \beta) \Cov(X)
     (\hat{\beta} - \beta)^T   & (\mbox{MSE Pred})\\
   \sum_{j=1}^p I(\hat{\beta}_j \neq 0, \beta_j \neq 0)/\sum_{j=1}^p
     I(\beta_j \neq 0) & (\mbox{true positive rate})\\
   \sum_{j=1}^p I(\hat{\beta}_j \neq 0, \beta_j = 0)/\sum_{j=1}^p
     I(\beta_j = 0) & (\mbox{false positive rate})
   \end{array}
\end{align}
where $I(\cdot)$ denotes the indicator function.

We apply the PC-simple algorithm for variable selection and then use the Lasso or the adaptive Lasso to
estimate the coefficients for the sub-model selected by the PC-simple
algorithm. We compare this procedure to the Lasso, the adaptive Lasso and the relaxed
Lasso. For simplicity, we do not show results for Elastic Net, since this method was found to
be worse in terms of receiver operating characteristic curves than the Lasso, see Section \ref{sec.ROC}.

Prediction optimal tuning is pursued with
a validation set having the same size as the training data. For the
adaptive Lasso, we first compute a prediction-optimal Lasso as initial
estimator $\hat{\beta}_{init}$, and
the adaptive Lasso is then computed by solving the following optimization
problem:
\begin{align*}
   \text{argmin}_{\beta\in \R^p} \left \{ \sum_{i=1}^n (Y_i-X_i^T\beta)^2 +
     \lambda \sum_{j=1}^p w_j^{-1} |\beta_j|\right \},
\end{align*}
where $w_j^{-1} = |\hat{\beta}_{init,j}|^{-1}$ and $\lambda$ is again
chosen in a prediction-optimal way. The computations are done with the
\texttt{R}-package \texttt{lars}, using re-scaled covariates for the
adaptive step. The relaxed Lasso is computed with the \texttt{R}-package
\texttt{relaxo}. The PC-simple algorithm with the Lasso for estimating
coefficients is computed using the \texttt{R}-packages \texttt{pcalg} and
\texttt{lars}, using optimal tuning with respect to the $\alpha$-parameter
for the PC-simple algorithm and the penalty parameter for Lasso. For the
PC-simple algorithm with the adaptive Lasso, we first compute weights $w_j$
as follows. If the variable has not been selected, we set $w_j = 0$. If the
variable has been selected, we let $w_j$ be the minimum value of the test
statistic $(n - 3 - |{\cal S}|)^{1/2} Z(Y,j \mid {\cal S})$ over all
iterations of the PC-simple algorithm. With these weights $w_j$, we then
compute the adaptive Lasso as defined above.

\begin{figure}[!htb]
  \centerline{
      \includegraphics[scale=0.4,angle=-90]{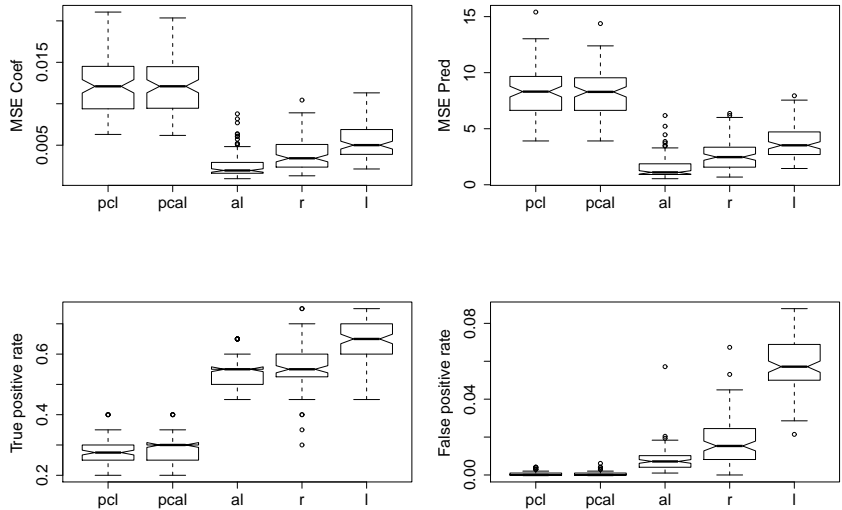}}
    \caption{Boxplots of performance measures for the simulation study in
      Section \ref{sec.opt.pred.tuned} considering the following prediction
      optimal tuned methods: the PC-simple algorithm with Lasso coefficient
      estimation (pcl), the PC-simple algorithm with adaptive Lasso (pcal),
      the adaptive Lasso (al), the relaxed Lasso (r) and the Lasso (l).}
  \label{fig:predopt}
\end{figure}

The results are shown in Figure \ref{fig:predopt}. As expected, the Lasso
yields too many false positives, while the adaptive Lasso and the relaxed
Lasso have much better variable selection properties. The PC-simple based
methods clearly have the lowest false positive rates, but pay a price in
terms of the true positive rate and mean squared errors. In many
applications, a low false positive rate is highly desirable even when
paying a price in terms of power. For example, in molecular biology where a
covariate represents a gene, only a limited number of selected genes
can be experimentally validated. Hence, methods with a low
false positive rate are preferred, in the hope that most of the
top-selected genes are relevant, as sketched in the next section.

\subsection{Real data: riboflavin production by Bacillus
  subtilis}\label{subsec.realdata}

We consider a high-dimensional real data set about vitamin riboflavin
production by the bacterium B. subtilis, kindly provided by DSM Nutritional
Products. There is a continuous response variable $Y$ which measures the
logarithm of the riboflavin production rate, and there are $p = 4088$
covariates corresponding to the logarithms of expression levels of
genes. The main goal is to genetically modify B. subtilis in order to
increase its riboflavin production rate. An important step to achieve this
goal is to find genes which are most relevant for the production rate.

We use data from a genetically homogeneous group of $n = 71$ individuals.
We run the PC-simple algorithm on the full data set for various values of $\alpha$.
Next, we compute the Lasso and Elastic Net, choosing
the tuning parameters such that they select the same number of variables as the
PC-simple algorithm.

Table \ref{tab:DSMdata} shows that there is overlap between the selected variables
of the three different methods. This overlap is highly significant when calibrating with a
null-distribution which consists of random noise. On the other hand,
we see that the variable selection results of
the Lasso and Elastic Net are more similar than the results of the PC-simple algorithm
and either of these methods. Hence, the
PC-simple algorithm seems to select genes in a different way than the penalty-based methods Lasso and Elastic Net.
This is a desirable finding, since for any large-scale problem, we want to
see different aspects of the problem by using different methods. Ideally,
results from different methods can then be combined to obtain results that are better
than what is achievable by a single procedure.

\begin{table}[htbp]
  \begin{center}
    \begin{tabular}{|l|c|c|c|c|}
      \hline
      $\alpha$ for PC-simple& Selected & PC-Lasso &
      PC-Enet & Lasso-Enet \\
      \hline
      0$\cdot$001 & 3 & 0 & 0 & 2  \\
      0$\cdot$01  & 4 & 2 & 1 & 3\\
      0$\cdot$05  & 5 & 2 & 1 & 3 \\
      0$\cdot$15 & 6  & 3 & 2 & 3\\
      \hline
    \end{tabular}
    \caption{Variable selection for a real data set on riboflavin
      production by B. subtilis. The columns 2 to 5 show the number of
      variables selected by the PC-simple algorithm, the number of variables
      selected by both the PC-simple algorithm and the Lasso,
      the number of variables selected
      by both the PC-simple algorithm and Elastic Net,
      and the number of variables that were selected by both the
      Lasso and Elastic Net.}
    \label{tab:DSMdata}
  \end{center}
\end{table}

\vspace*{-5mm}
\appendix
\section*{Appendix}\label{sec.proofs}
\subsection*{Proofs}

\begin{proof}[of Theorem~\ref{theorem.mod.lin.part.faithful}]
   Consider the linear model \eqref{mod.lin} satisfying assumptions (C1) and (C2). In order to prove that the partial faithfulness assumption holds almost surely, it suffices to show that the following holds for all $j\in\{1,\dots,p\}$ and $S\subseteq \{j\}^C$:
      $\beta_j \neq 0$  implies that $\rho(Y,X^{(j)} \mid X^{(\cal S)})
      \neq  0$ almost surely
   with respect to the distribution generating the $\beta_j$s.

   Thus, let $j\in \{1,\dots,p\}$ such that $\beta_j \neq 0$, and let $S
   \subseteq \{j\}^C$. We recall that $\rho(Y,X^{(j)} \mid X^{(\cal S)})
   =0$ if and only if the partial covariance $\parcov(Y, X^{(j)} \mid
   X^{(\cal S)})$ between $Y$ and $X^{(j)}$ given $X^{(\cal S)}$ equals
   zero as can be seen in \citet[page 37, definition
   2$\cdot$5$\cdot$2]{anderson84}. Partial covariances can be computed
   using the recursive formula given in \citet[page 43, equation
   (26)]{anderson84}. This formula shows that the partial covariance is
   linear in its arguments, and that
$\parcov(\epsilon,X^{(j)} \mid X^{(\cal S)})=0$ for all $j \in \{1,\ldots ,p\}$
and ${\cal S} \subseteq \{j\}^C$.
 Hence,
   \begin{align*}
      \parcov(Y, X^{(j)} \mid X^{(\cal S)}) & = \parcov(\delta +
      \sum_{r=1}^p \beta_r X^{(r)} + \epsilon, X^{(j)}  \mid  X^{(\cal S)}) \\
      & = \sum_{r=1}^p \beta_r \parcov(X^{(r)}, X^{(j)} \mid X^{(\cal S)})
      \\
& = \beta_j \parcov(X^{(j)},X^{(j)} \mid X^{(\cal S)}) + \sum_{r=1, r\neq j}^p \beta_r \parcov(X^{(r)},X^{(j)} \mid X^{(\cal S)}).
   \end{align*}
   Since $\beta_j\neq 0$ by assumption, and since $\parcov(X^{(j)},X^{(j)} \mid X^{(\cal S)}) \neq 0$ by assumption (C1), the only way for $\parcov(Y,X^{(j)} \mid X^{(\cal S)})$ to equal zero is if there is a special parameter configuration of the $\beta_r$s, such that
   \begin{align}
       \sum_{r=1, r\neq j}^p \beta_r \parcov(X^{(r)},X^{(j)} \mid X^{(\cal S)}) = - \beta_j \parcov(X^{(j)},X^{(j)} \mid X^{(\cal S)}).
       \label{eq.parameter.cancellation}
   \end{align}
But such a parameter constellation has Lebesgue measure zero under
assumption (C2).
\end{proof}

\begin{proof}[of Corollary \ref{corollary.all.par.corr.neq.0}]
   The implication from the left to the right hand side follows from the
   fact that $\beta_j\neq 0$ in the linear model \eqref{mod.lin} if and
   only if
   $\rho(Y, X^{(j)} \mid X^{(\{j\}^C)}) \neq 0$. The other direction follows from the definition of partial
   faithfulness, by
   taking the negative of expression \eqref{eq: def partial faithfulness for mod.lin}.
\end{proof}

\begin{proof}[of Theorem \ref{theorem: faithfulness and part faithfulness}]
   Suppose that $(X,Y) = (X^{(1)},\dots,X^{(p)},Y)$ is linearly $Y$-faithful to a directed acyclic graph
   $G$ in which $Y$ is childless, i.e., any edges between $Y$ and the $X^{(j)}$s, $j=1,\dots,p$,
   point towards $Y$. We will show that this implies that the
   distribution of $(X,Y)$ is partially
   faithful, by showing that $\rho(Y,X^{(j)} \mid X^{(\{j\}^C)})\neq 0$ implies that $\rho(Y,X^{(j)} \mid X^{(\cal S)})\neq 0$ for all ${\cal S}\subseteq \{j\}^C$.

   Thus, let $j\in \{1,\dots,p\}$ such that $\rho(Y,X^{(j)} \mid X^{(\{j\}^C)})\neq 0$. By linear $Y$-faithfulness, this implies that $Y$ and $X^{(j)}$ are not d-separated by $X^{(\{j\}^C)}$ in $G$, meaning that $X^{( \{j\}^C)}$ does not block all d-connecting paths between $X^{(j)}$ and $Y$. All paths between $X^{(j)}$ and $Y$ must be of the form $X^{(j)} - \dots - \dots \dots - X^{(r)} \rightarrow Y$, where $-$ denotes an edge of the form $\leftarrow$ or $\rightarrow$. First suppose that $r\neq j$. Then, because
   $X^{(r)}$ cannot be a collider on the given path (since we know that the edge from $X^{(r)}$ to $Y$ points towards $Y$), the path is blocked by $X^{(r)}\in X^{(\{j\}^C)}$, and hence the path is blocked by $X^{(\{j\}^C)}$. Thus, since $X^{( \{j\}^C)}$ does not block all paths between $X^{(j)}$ and $Y$, there must be a path where $r=j$, or in other words, there must be an edge between $X^{(j)}$ and $Y$: $X^{(j)} \to Y$. Such a path $X^{(j)} \to Y$ cannot be blocked by any set $X^{(\cal S)}$, ${\cal S} \subseteq \{j\}^C$. Hence, there does not exist a set ${\cal S}$ that d-separates $X^{(j)}$ and $Y$. By linear $Y$-faithfulness, this implies that $\rho(X^{(j)},Y \mid X^{({\cal S})})\neq 0$ for all ${\cal S} \subseteq \{j\}^C$.
\end{proof}

\begin{proof}[of Theorem \ref{theorem.correctness.pc.simple.pop}]
   By partial faithfulness and equation \eqref{screen-models}, ${\cal A} \subseteq {\cal A}^{[m_{\reach}]}$. Hence, we only need to show that ${\cal A}$ is not a strict subset of ${\cal A}^{[m_{\reach}]}$.
   We do this using contra-position. Thus, suppose that ${\cal A} \subset {\cal A}^{[m_{\reach}]}$ strictly.
   Then there exists a $j \in {\cal A}^{[m_{\reach}]}$ such that $j\notin
   {\cal A}$. Fix such an index $j$. Since $j\in {\cal A}^{[m_{\reach}]}$,
   we know that
   \begin{align}\label{pc-tested0}
     \rho(Y,X^{(j)} \mid X^{({\cal S})}) \neq 0\ \mbox{for all}\ {\cal S}
       \subseteq {\cal A}^{[m_{\reach}-1]} \setminus \{j\}
       \quad \text{with} \quad |{\cal S}| \le m_{\reach}-1.
   \end{align}
   This statement for sets $\cal S$ with $|{\cal S}|=m_{\reach}-1$ follows from the definition of iteration $m_{\reach}$ of the PC-simple algorithm. Sets $\cal S$ with lower cardinality are considered in previous iterations of the algorithm, and since $A^{[1]} \supseteq {\cal A}^{[2]} \supseteq \ldots$, all subsets ${\cal S} \subseteq {\cal A}^{[m_{\reach} -1]}$ with $|{\cal S}| \le m_{\reach} - 1$ are considered.

   We now show that we can take ${\cal S} = {\cal A}$ in
     (\ref{pc-tested0}). First, the supposition ${\cal A} \subset {\cal A}^{[m_{\reach}]}$ and our choice of $j$ imply that
   $${\cal A} \subseteq {\cal A}^{[m_{\reach}]} \setminus \{j\} \subseteq {\cal A}^{[m_{\reach}-1]} \setminus \{j\}.$$
   Moreover, ${\cal A} \subset {\cal A}^{[m_{\reach}]}$ implies that
   ${|\cal A|} \le  |{\cal A}^{[m_{\reach}]}| -1 $. Combining this with
   $|{\cal A}^{[m_{\reach}]}| \le m_{\reach}$ yields that ${|\cal A|} \le
   m_{\reach}-1$. Hence, we can indeed take $\cal S = \cal A$ in
   \eqref{pc-tested0}, yielding that $\rho(Y,X^{(j)} \mid X^{(\cal A)}) \neq
   0$.

   On the other hand, $j\notin \cal A$ implies that $\beta_j=0$, and hence
   $\rho(Y,X^{(j)} \mid X^{(\cal A)})=0$. This is a contradiction, and hence
   $\cal A$ cannot be a strict subset of ${\cal
     A}^{[m_{\reach}]}$.
\end{proof}

\begin{proof}[of Theorem \ref{theorem.cons.pc.simple}]
  A first main step is to show that the population version of the PC-simple
  algorithm infers the true underlying active set ${\cal A}_n$, assuming
  partial faithfulness. We formulated this step as a separate result in Theorem
  \ref{theorem.correctness.pc.simple.pop}.

  The arguments for controlling the estimation error due to a finite sample
  size are similar to the ones used in the proof of Theorem 1 in
  \citet{kabu07}. We proceed in two steps, analyzing first partial
  correlations and then the PC-simple algorithm.

\medskip
   We show an exponential inequality for
   estimating partial correlations up to order $m_n = o(n)$. We use the
   following notation: $K_{j}^{m_n} = \{{\cal S} \subseteq \{0,\ldots
   ,p_n\} \setminus \{j\}; |{\cal S}| \le m_n\}$ ($j=1,\ldots ,p_n$). We
   require more general versions of assumptions (D4) and (D5) where the
   cardinality of the condition sets are bounded by the number $m_n$:
\begin{enumerate}
\item[(D4$_{m_n}$)] The partial correlations $\rho_n(Y,j \mid {\cal S}) =
  \rho(Y,X^{(j)} \mid X^{({\cal S})})$ satisfy:
\begin{eqnarray*}
  & &\inf\Big\{|\rho_n(Y,j \mid {\cal S})|;\ j=1,\ldots ,p_n,\ {\cal S}
  \subseteq \{j\}^C,\ |{\cal S}| \le m_n\ \mbox{with}\
  \rho_n(Y,j \mid {\cal S}) \neq 0\Big\} \ge c_n,
\end{eqnarray*}
where $c_n^{-1} = O(n^{d})$ for some $0 \le d < b/2$, and $b$ is as in
(D3).
\item[(D5$_{m_n}$)] The partial correlations $\rho_n(Y,j \mid {\cal S})$ and
  $\rho_n(i,j \mid {\cal S}) = \rho(X^{(i)},X^{(j)} \mid X^{({\cal S})})$ satisfy:
\begin{eqnarray*}
(i)\ \sup_{n,j,{\cal S} \subseteq \{j\}^C,|{\cal S}| \le m_n} |\rho_n(Y,j \mid {\cal S})| \le M < 1,\ \ (ii)\ \sup_{n,i\neq j,{\cal S} \subseteq
  \{i,j\}^C,|{\cal S}| \le m_n} |\rho_n(i,j \mid {\cal S})| \le M < 1.
\end{eqnarray*}
\end{enumerate}
We will later see in Lemma
\ref{lemma:mreach} that we need $m_n \le \peff_n$ only, and hence,
assumptions (D4$_{m_n}$) and (D5$_{m_n}$) coincide with (D4) and (D5),
respectively.

We have,
   for $m_n < n-4$ and
   $0<\gamma<2$,
   \begin{eqnarray*}
     \sup_{{\cal S} \in K_{j}^{m_n},j=1,\ldots ,p_n}
   \PP\{|\hat{\rho}_n(Y,j \mid {\cal S}) - \rho_n(Y,j \mid {\cal S})| >
   \gamma\}\nonumber
    \le  C_1 n
   \exp(n - m_n - 4) \log\left(\frac{4 - \gamma^2}{4+\gamma^2}\right),
   \end{eqnarray*}
   where $0 < C_1 < \infty$ depends on $M$ in (D5$_{m_n}$) only. This bound
   appears
   in \citet[Corollary 1]{kabu07}: for proving it, we require the Gaussian
   assumption for the distribution and (D5$_{m_n}$). It is now straightforward to
   derive an exponential inequality for the estimated $Z$-transformed
   partial correlations. We define $Z_{n}(Y,j \mid {\cal S}) = g\{\hat
   \rho_n(Y,j \mid {\cal S})\}$ and $z_{n}(Y,j \mid {\cal S}) =
   g\{\rho_{n}(Y,j \mid {\cal S})\}$, where $g(\rho) = \frac{1}{2}
   \log\{(1+\rho)/(1-\rho)\}$.
\begin{lemma}\label{lemma:3a}
Suppose that the Gaussian assumption from (D1) and condition (D5$_{m_n}$) hold.
Define $L = 1/\{1 - (1+M)^2/4\}$, with $M$ as in assumption (D5$_{m_n}$).
Then, for $m_n<n-4$ and $0<\gamma<2L$,
\begin{eqnarray*}
& &\sup_{{\cal S} \in K_{j}^{m_n}, j=1,\ldots ,p_n}
\PP\{|Z_n(Y,j \mid {\cal S})-z_n(Y,j \mid {\cal S})|>\gamma\} \\
&\le& O(n) \left(\exp \left [(n-4 -m_n)
  \log \left \{\frac{4-(\gamma/L)^2}{4+(\gamma/L)^2} \right \} \right] + \exp\{-C_2(n-m_n)\}\right)
\end{eqnarray*}
for some constant $C_2>0$.
\end{lemma}
We omit the proof since this is Lemma 3 in \citet{kabu07}.

\medskip
We now consider a version of the PC-simple algorithm that stops after a
fixed number of $m$ iterations. If $m \ge \hat{m}_{\reach}$, where $\hat
m_{\reach}$ is the estimation analogue of (\ref{reach}), we set
$\widehat{\cal A}^{[m]} = \widehat{A}^{[\hat{m}_{\reach}]}$. We denote this
version by PC-simple($m$) and the resulting estimate by $\widehat{\cal
  A}(\alpha,m)$.

\begin{lemma}\label{lemma:Gskel}
Assume (D1)-(D3), (D4$_{m_n}$) and (D5$_{m_n}$).
Then, for $m_n$ satisfying $m_n \ge
m_{\reach,n}$ and $m_n = O(n^{1-b})$ with $b$ as in (D3), there
exists a sequence $\alpha_n \to 0$ such that
\begin{eqnarray*}
\PP\{\widehat{\cal A}_n(\alpha_n,m_n) = {\cal A}_n\} = 1 - O\{\exp(-Cn^{1 -
  2d})\} \to 1\ (n \to \infty)\ \mbox{for some}\ C >0.
\end{eqnarray*}
\end{lemma}
A concrete choice of $\alpha_n$ is $\alpha_n = 2\{1 - \Phi(n^{1/2} c_n/2)\}$, where $c_n$ is the lower
bound from (D4$_{m_n}$), which is typically unknown.
\begin{proof}
  Obviously, the population version of the PC-simple($m_n$) algorithm is
  correct for $m_n \ge m_{\reach,n}$, see Theorem
  \ref{theorem.correctness.pc.simple.pop}. An error can occur in the
  PC-simple($m_n$) algorithm if there exists a covariate $X^{(j)}$ and a
  conditioning set ${\cal S} \in K_{j}^{m_n}$ for which an error event $E_{j
    \mid {\cal S}}$ occurs, where $E_{j \mid {\cal S}}$ denotes the event that an
  error occurred when testing $\rho_n(Y,j \mid {\cal S}) = 0$. Thus,
\begin{eqnarray}\label{err-pc}
& &\PP\{\mbox{an error occurs in the PC-simple($m_n$)-algorithm}\}\nonumber \\
&\le&
\PP\bigg(\bigcup_{{\cal S} \in K_{j}^{m_n},j=1,\ldots ,p_n} E_{j \mid {\cal S}}\bigg)
\le O(p_n^{m_n+1}) \sup_{{\cal S} \in K_{j}^{m_n},j}
\PP(E_{j \mid {\cal S}}),
\end{eqnarray}
using that the cardinality of the index set $\{{\cal S} \in
K_{j}^{m_n},j=1,\ldots ,p_n\}$ in the union is bounded by $O(p_n^{m_n+1})$.
Now
\begin{eqnarray}\label{ADD1}
E_{j \mid {\cal S}} = E_{j \mid {\cal S}}^I \cup E_{j \mid {\cal S}}^{II},
\end{eqnarray}
where
\begin{eqnarray*}
& &\mbox{type I error}\ E_{j \mid {\cal S}}^I:\ (n - |{\cal S}| - 3)^{1/2}
   |Z_{n}(Y,j \mid {\cal S})| > \Phi^{-1} (1- \alpha/2)\ \mbox{and}\
   z_{n}(Y,j \mid {\cal S}) =0,\\
& &\mbox{type II error}\ E_{j \mid {\cal S}}^{II}:\ (n - |{\cal S}| - 3)^{1/2}
   |Z_{n}(Y,j \mid {\cal S})| \le \Phi^{-1}
   (1- \alpha/2)\ \mbox{and}\ z_{n}(Y,j \mid {\cal S}) \neq 0.
\end{eqnarray*}
Choose $\alpha = \alpha_n = 2\{1 - \Phi(n^{1/2} c_n/2)\}$, where $c_n$ is from
(D4$_{m_n}$). Then,
\begin{eqnarray}\label{ADD2}
\sup_{{\cal S} \in K_{j}^{m_n},j=1,\ldots ,p_n}  \PP(E_{j \mid {\cal S}}^I) & = &
\sup_{{\cal S} \in K_{j}^{m_n},j}  \PP\left[|Z_{n}(Y,j \mid {\cal S}) -
z_{n}(Y,j \mid {\cal S})| > \{n/(n - |{\cal S}| - 3)\}^{1/2} c_n/2\right]\nonumber\\
&\le & O(n) \exp\{-C_3 (n - m_n) c_n^2\},
\end{eqnarray}
for some $C_3>0$,
 using Lemma \ref{lemma:3a} and the fact that
$\log\{(4-\delta^2)/(4+\delta^2)\} \le -\delta^2/2$ for $0<\delta<2$.
Furthermore, with the choice of $\alpha = \alpha_n$ above,
\begin{eqnarray*}
& &\sup_{{\cal S} \in K_{j}^{m_n},j=1,\ldots ,p_n}  \PP(E_{j \mid {\cal S}}^{II}) =
\sup_{{\cal S} \in K_{j}^{m_n}, j } \PP\left[|Z_{n}(Y,j \mid {\cal S})| \le \{n/(n
- |{\cal S}| - 3)\}^{1/2} c_n/2\right] \\
&\le& \sup_{{\cal S} \in K_{j}^{m_n},j} \PP\left(|Z_{n}(Y,j \mid {\cal S}) -
z_{n}(Y,j \mid {\cal S})| > c_n[1 - \{n/(n - |{\cal S}| - 3)\}^{1/2}/2]\right),
\end{eqnarray*}
because $\inf_{{\cal S} \in K_{j}^{m_n},j}\{|z_{n}(Y,j \mid {\cal S})|;\
z_n(Y,j \mid {\cal S}) \neq 0\} \ge
c_n$ since $|g(\rho)| = |\frac{1}{2} \log\{(1+\rho)/(1-\rho)\}| \ge
|\rho|$ for all $\rho$ and using assumption
(D4$_{m_n}$). This shows the crucial role of assumption (D4$_{m_n}$) in
controlling the type II
error. By invoking Lemma \ref{lemma:3a} we then obtain:
\begin{eqnarray}\label{ADD3}
\sup_{{\cal S} \in K_{j}^{m_n},j}  \PP(E_{j \mid {\cal S}}^{II}) \le
O(n) \exp\{-C_4 (n - m_n) c_n^2\}
\end{eqnarray}
for some $C_4 >0$.
Now, by (\ref{err-pc})-(\ref{ADD3}) we get
\begin{eqnarray*}
& &\PP\{\mbox{an error occurs in the PC-simple($m_n$)-algorithm}\}\\
&\le& O[p_n^{m_n+1} n \exp\{-C_5 (n - m_n) c_n^2\}] \le O[n^{a (m_n+1) +1}
\exp\{-C_5 (n - m_n) n^{-2d}\}] \\
&= &O\left[\exp \left\{a(m_n+1) \log(n) + \log(n) -
C_5 (n^{1 - 2d} - m_n n^{-2d})\right\}\right] = o(1),
\end{eqnarray*}
because $n^{1 - 2d}$ dominates all other terms in the argument of the
$\exp$-function, due to $m_n=O(n^{1-b})$ and the assumption in (D4$_{m_n}$)
that $d<b/2$. This completes
the proof.
\end{proof}

\medskip
Lemma \ref{lemma:Gskel} leaves some flexibility for choosing $m_n$. The
PC-algorithm yields a data-dependent stopping level
$\hat{m}_{\reach,n}$, that is, the sample version of (\ref{reach}).

\begin{lemma}\label{lemma:mreach}
Assume (D1)-(D5). Then,
\begin{eqnarray*}
& &\PP(\hat{m}_{\reach,n} = m_{\reach,n}) = 1 - O\{\exp(-Cn^{1-2d})\} \to 1\
(n \to \infty)
\end{eqnarray*}
for some $C >0$, with $d$ is as in (D4).
\end{lemma}

\begin{proof}
Consider the population version of the PC-simple algorithm, with stopping
level $m_{\reach}$ as defined in (\ref{reach}). Note that $m_{\reach} = m_{\reach,n}=O(n^{1-b})$ under assumption (D3).
The sample PC-simple($m_n$) algorithm
with stopping level in the range of $m_n \ge m_{\reach}\ (m_n = O(n^{1-b}))$,
coincides
with the population version on a set $A$ having probability $P[A] = 1 -
O\{\exp(-Cn^{1-2d})\}$, see the last formula in the proof of Lemma
\ref{lemma:Gskel}. Hence, on the set $A$, $\hat{m}_{\reach,n} =
m_{\reach}$.
\end{proof}


Lemma \ref{lemma:Gskel} with $m_n = \peff_n$ together with Lemma
\ref{lemma:mreach}, and using that $m_{\reach,n} \le \peff_n$,
complete the proof of Theorem \ref{theorem.cons.pc.simple}.
\end{proof}

\begin{proof}[of Theorem \ref{theorem.cons.corr.screening}]
   By definition, ${\cal A}_{n} \subseteq {\cal A}^{[1]}$ for the
   population version. Denote by $Z_n(Y,j)$ the quantity as in
   (\ref{ztrans}) with ${\cal S} =
   \emptyset$ and by $z_n(Y,j)$ its population analogue, i.e., the
   $Z$-transformed population correlation.
   An error occurs when screening the $j$th variable if $Z_n(Y,j)$ has
   been tested to be zero but in fact $z_n(Y,j) \neq 0$. We denote such an
   error event by $E^{II}_j$. Note that
   \begin{eqnarray*}
      \sup_{j=1,\ldots ,p_n} \PP(E^{II}_j) \le O(n) \exp(-C_1 n c_n^2),
   \end{eqnarray*}
   for some $C_1 >0$, see formula (\ref{ADD3}) above. We do not use any
   sparsity assumption for this derivation, but we do invoke (E1) which
   requires a lower bound on non-zero marginal correlations. Thus, the
   probability of an error occurring in the correlation screening procedure
   is bounded: for some $C_2 >0$,
   \begin{eqnarray*}
      \PP\big(\cup_{j=1,\ldots ,p_n} E^{II}_j\big) &=& O(p_n n) \exp(-C_1 n c_n^2) =
      O[\exp\{(1+a) \log(n) - C_1 n^{1 - 2d}\}]\\
      &=& O\{\exp(-C_2 n^{1 - 2d})\}.
   \end{eqnarray*}
\end{proof}

\bibliographystyle{biometrika}
\bibliography{faithful}

\end{document}